\documentclass[12pt]{article}

\usepackage[margin=1in]{geometry}

\usepackage{setspace}

\usepackage{graphicx}

\usepackage{enumitem}

\usepackage{microtype}
\usepackage{subfigure}
\usepackage{booktabs}
\usepackage{comment}
\usepackage{wrapfig}
\usepackage{arydshln}
\usepackage{enumitem}
\usepackage{multirow}
\usepackage{url}
\usepackage{natbib}
\usepackage{authblk}

\RequirePackage[colorlinks,citecolor=blue,linkcolor=blue,urlcolor=blue,pagebackref]{hyperref}

\usepackage{amsmath}
\usepackage{amssymb}
\usepackage{mathtools}
\usepackage{amsfonts}
\usepackage{bm}
\usepackage{bbm}

\usepackage{algorithm}
\usepackage{algorithmic}

\def\eqref#1{equation~\ref{#1}}

\def\1{\bm{1}}

\DeclareMathAlphabet{\mathsfit}{\encodingdefault}{\sfdefault}{m}{sl}
\SetMathAlphabet{\mathsfit}{bold}{\encodingdefault}{\sfdefault}{bx}{n}

\def\calF{{\mathcal{F}}}
\def\calG{{\mathcal{G}}}
\def\calH{{\mathcal{H}}}

\def\calM{{\mathcal{M}}}

\def\calQ{{\mathcal{Q}}}
\def\calR{{\mathcal{R}}}

\def\calT{{\mathcal{T}}}

\def\calX{{\mathcal{X}}}
\def\calY{{\mathcal{Y}}}
\def\calZ{{\mathcal{Z}}}

\def\bbE{{\mathbb{E}}}

\def\bbN{{\mathbb{N}}}

\def\bbP{{\mathbb{P}}}

\def\bbR{{\mathbb{R}}}

\DeclareMathOperator*{\argmin}{arg\,min}

\newcommand{\p}[1]{\left(#1\right)}
\newcommand{\sqb}[1]{\left[#1\right]}
\newcommand{\cb}[1]{\left\{#1\right\}}

\newcommand{\Bigp}[1]{\Big(#1\Big)}

\newcommand{\abs}[1]{\left|#1\right|}

\newcommand{\norm}[1]{\left\|#1\right\|}

\usepackage{amsthm}

\theoremstyle{plain}

\newtheorem{theorem}{Theorem}[section]

\newtheorem{proposition}[theorem]{Proposition}

\newtheorem{assumption}{Assumption}[section]

\newtheorem*{remark}{Remark}

\usepackage[textsize=tiny]{todonotes}
\usepackage{multirow}
\usepackage{wrapfig}
\usepackage{subfigure}

\renewcommand{\eqref}[1]{(\ref{#1})}

\usepackage{xcolor}
\newcount\Comments  
\newcommand{\kibitz}[2]{\ifnum\Comments=1\textcolor{#1}{#2}\fi}

\usepackage{listings}

\usepackage[capitalize,noabbrev]{cleveref}

\allowdisplaybreaks

\title{Conformal Prediction for Nonparametric Instrumental Regression}

\author{Masahiro Kato\thanks{Email: \texttt{mkato-csecon@g.ecc.u-tokyo.ac.jp}}$\,$}

\affil{The University of Tokyo}

\date{\today}

\begin{document}

\maketitle 

\begin{abstract}
We propose a method for constructing distribution-free prediction intervals in nonparametric instrumental variable regression (NPIV), with finite-sample coverage guarantees. Building on the conditional guarantee framework in conformal inference, we reformulate conditional coverage as marginal coverage over a class of IV shifts $\calF$. Our method can be combined with any NPIV estimator, including sieve 2SLS and other machine-learning-based NPIV methods such as  neural networks minimax approaches. Our theoretical analysis establishes distribution-free, finite-sample coverage over a practitioner-chosen class of IV shifts.
\end{abstract}

\section{Introduction}
Instrumental variables (IVs) are widely used for causal and policy evaluation when regressors are endogenous. We consider structural prediction problems characterized by moment conditions conditional on IVs, for which a representative example is nonparametric IV regression \citep[NPIV,][]{Newey2003instrumentalvariable,Ai2003efficientestimation,Darolles2011nonparametricinstrumental}. A large literature studies how to estimate $h_0$ under ill-posedness, using series, minimum-distance, reproducing kernel Hilbert space (RKHS), neural network, and minimax procedures \citep{Newey2003instrumentalvariable,Hartford2017deepiv,Singh2019kernelinstrumental,Dikkala2020minimaxestimation}.

Here, we briefly introduce the NPIV setup. For details, see Section~\ref{sec:setup}. In NPIV, we typically consider the following data-generating process:
\begin{align}\label{eq:npiv}
Y = h_0(X) + \varepsilon, \qquad \bbE\sqb{\varepsilon \mid Z} = 0,
\end{align}
where $X$ is an endogenous variable, $\varepsilon$ is the error term, $Z$ is an IV, and $h_0(\cdot)$ is called the structural function. While the structural function $h_0(\cdot)$ relates the endogenous variable $X$ to the outcome $Y$, it cannot be estimated by regressing $Y$ on $X$ because of endogeneity, that is, the correlation between $X$ and $\varepsilon$.

In this paper, we study prediction intervals for future outcomes in NPIV. Existing inference methods in NPIV typically target the structural function $h_0$ or its functionals through asymptotic approximations to regularized estimators. Our goal is different. We aim to construct prediction intervals for $Y$ itself, with validity formulated relative to the IV. For example, consider predicting future demand for a good. In conventional approaches, one first estimates the demand curve $h_0$ and then predicts demand $Y$. In that case, one constructs confidence or prediction intervals for $h_0$. In this study, by contrast, we target $Y$ itself, that is, we construct prediction intervals for $Y$, not $h_0$. Our prediction intervals for $Y$ are justified by the IV $Z$.

Conformal prediction provides a natural starting point for this purpose, since it yields distribution-free finite-sample coverage under exchangeability \citep{Vovk2005algorithmiclearning}. While conformal prediction is a convenient approach in many tasks, it is not immediately applicable to NPIV, because exact conditional coverage is impossible \citep{Lei2013distributionfree,Barber2020thelimits}. That is, exact conditional coverage given the IV is impossible in a fully distribution-free, finite-sample sense. To address this issue, we build on the results in \citet{Tibshirani2019conformalprediction} and \citet{Gibbs2025conformalprediction}, where the former considers conformal prediction under covariate shift, and the latter considers conditional guarantees in conformal prediction.

A basic modeling question then arises. In IV problems, what should the final interval depend on? The structural regression itself is a function of $X$ alone, which suggests an $X$-indexed prediction interval $C(X)$ as the most natural predictive object. At the same time, one may also allow the radius to vary jointly with $(X,Z)$ or only with $Z$. These three possibilities have different interpretations and different algorithmic properties. We study conformal procedures for all three cases.

Within that common target, the three choices of radius lead to different methods. When the radius depends jointly on $(X,Z)$, the finite-dimensional conditional conformal machinery of \citet{Gibbs2025conformalprediction} applies directly with the covariate $(X,Z)$, yielding exact finite-sample coverage over joint shifts in $(X,Z)$. When the radius depends only on $Z$, the same machinery specializes to an exact IV-specific procedure, which we call IV-conditional conformal prediction, with coverage over IV shifts. When the radius depends only on $X$, exact finite-sample family-wise calibration is not currently available. To handle that case, we use the importance-weighting logic of \citet{Kato2022learningcausal} to convert the conditional coverage moment into weighted unconditional moments, and then combine that idea with weighted conformal recalibration under a fixed target shift in the spirit of \citet{Tibshirani2019conformalprediction}.

\subsection{Contribution}
We propose methods for prediction intervals in NPIV with distribution-free finite-sample coverage guarantees. In particular, we refer to the IV-specific conformal construction as \emph{IV-CCP}. IV-CCP is based on conformal prediction, in particular, on conditional guarantees through conditional calibration.

The methodological contribution is to transform the NPIV problem with conditional moment restrictions into one under marginal moment restrictions so that we can apply conformal prediction. We recast exact conditional coverage as a moment condition against all measurable reweightings of the IV and relax that impossible target to robust coverage over a user-specified class of IV shifts $\calF$. This approach is closely related to covariate shift \citep{Shimodaira2000improvingpredictive}.

In implementation, we use an estimate of $h_0(x)$ to tighten the prediction intervals. As a basic construction, we first estimate the structural function $h_0(x)$ and then construct prediction intervals by adding a radius to the estimator $\widehat h$. We study three types of radii. The first depends on both $X$ and $Z$, the second depends only on $Z$, and the third depends only on $X$. For each approach, we present a concrete implementation and theoretical guarantees.

\subsection{Related Work}
Our work lies at the intersection of NPIV estimation and conformal prediction. Estimation of causal parameters under conditional moment restrictions is a core problem in causal inference, and NPIV is a representative example. In this problem, the technical difficulties come from identification and stable estimation, including ill-posedness and regularization for $h_0$ \citep{Newey2003instrumentalvariable,Ai2003efficientestimation,Darolles2011nonparametricinstrumental}. In the machine learning literature, several studies use flexible models for the structural function based on neural networks and RKHS, or employ a minimax approach to estimation \citep{Hartford2017deepiv,Singh2019kernelinstrumental,Dikkala2020minimaxestimation}. Our method is agnostic to the particular NPIV models and estimation methods and therefore acts as a predictive wrapper around these estimators rather than a replacement for them.

On the conformal side, the starting point is distribution-free finite-sample validity under exchangeability \citep{Vovk2005algorithmiclearning}. Exact distribution-free conditional coverage is impossible in general \citep{Lei2013distributionfree,Barber2020thelimits}, so recent work has studied structured relaxations, including weighted conformal prediction under a fixed target shift \citep{Tibshirani2019conformalprediction}, localized conformal prediction \citep{Guan2022localizedconformal}, distributional conformal prediction \citep{Chernozhukov2021distributionalconformal}, and conditional guarantees over a family of shifts \citep{Gibbs2025conformalprediction}. The present paper specializes these ideas to IV settings and combines the conformal perspective with the importance-weighted conditional-moment formulation of \citet{Kato2022learningcausal}. Appendix~\ref{appdx:detailed_related_work} provides a fuller discussion.

\section{Setup}
\label{sec:setup}
In this section, we formulate our problem.

\subsection{Variables and Observations}
\paragraph{Variables}
Let $Y \in \calY \subseteq \bbR$ be a scalar outcome, $X \in \calX \subseteq \bbR^{k_X}$ be a $k_X$-dimensional endogenous variable, and $Z \in \calZ \subseteq \bbR^{k_Z}$ be $k_Z$-dimensional IVs, where $\calY$, $\calX$, and $\calZ$ denote the outcome, covariate, and IV spaces, respectively. Assume that $(Y, X, Z)$ jointly follows a distribution $P$.

\paragraph{Observations}
Let $n \in \bbN$ be the sample size. Assume that we observe a dataset $\cb{(Y_i,X_i,Z_i)}^n_{i=1}$, where each $(Y_i, X_i, Z_i)$ is an i.i.d. copy of $(Y, X, Z)$ following the distribution $P$.

\subsection{Prediction under Endogeneity}
This study considers prediction of $Y$ under endogeneity of $X$. We clarify the meaning of endogeneity below.

\paragraph{Recap of identification through conditional moment restrictions}
The conventional approach defines the structural function through identification under conditional moment restrictions \citep{Ai2003efficientestimation,Darolles2011nonparametricinstrumental}. In this approach, we define the causal relationship between $Y$ and $X$ as
\begin{align*}
Y = h_0(X) + \varepsilon,
\end{align*}
where $h_0:\calX\to\calY$ is the structural function and $\varepsilon$ is a sub-Gaussian error term with mean zero. To estimate $h_0$, suppose that the IV $Z$ satisfies the following conditional moment restriction:
\begin{align}
\label{eq:cmc}
\bbE\sqb{\varepsilon_i \mid Z_i} = 0,\quad \forall i \in \cb{1,2,\dots, n}.
\end{align}
Then we assume that $h_0$ is uniquely identified under \eqref{eq:cmc}. The goal of the conventional setup is to estimate $h_0$ from the conditional moment restriction in \eqref{eq:cmc}. If $Z = X$, this problem reduces to estimation of the regression function $\bbE\sqb{Y \mid X}$. However, when $\bbE\sqb{\varepsilon\mid X} \neq 0$, $\bbE\sqb{Y\mid X}$ is not equal to $h_0(X)$, so standard regression methods such as least squares may fail to recover the structural function. In such cases, we use NPIV methods to estimate $h_0$ using IVs $Z$.

\paragraph{Reformulation by conditional prediction given IVs}
The conventional approach relies on identification of the structural function itself. In this study, we do not focus on identification. Instead, we consider prediction of $Y$ using $X$, with validity indexed by the IV $Z$. In this approach, we define the target as a prediction interval
\[\widehat C\colon \calX\times\calZ\to \cb{\text{intervals in }\bbR},\]
which satisfies
\begin{align}
\label{eq:exact_conditional0}
\bbP\p{Y_{n+1}\in \widehat C(X_{n+1},Z_{n+1})\mid Z_{n+1}=z}=1-\alpha,\quad \forall z \in \calZ.
\end{align}
We first define a prediction interval $\widehat C$ that depends on both $X$ and $Z$. We then introduce prediction intervals that depend only on $Z$ or only on $X$. Our main focus is the prediction interval that depends only on $X$. Note that in our proposed method, we also use an estimate of $h_0$ to tighten the prediction intervals, but this is not necessary.

\subsection{Ideal (but Infeasible) Goal: Prediction Interval with Conditional Coverage Guarantee}
In summary, our ideal goal is to predict $Y$ given $X$ while accounting for endogeneity and to obtain a prediction interval $\widehat C\colon \calX\times\calZ\to \cb{\text{intervals in }\bbR}$ satisfying the conditional coverage guarantee \eqref{eq:exact_conditional0}.

For the construction of $\widehat C(X_{n+1},Z_{n+1})$, we allow the use of an estimator $\widehat h\colon \calX\to\calY$ of the structural function $h_0$. This reflects the common application setting in which $X$ is the argument of the prediction rule and $Z$ is used to calibrate validity, not necessarily to serve as a direct input to the point predictor.

In this study, we aim to guarantee distribution-free and finite-sample coverage for the prediction interval. However, distribution-free procedures cannot satisfy \eqref{eq:exact_conditional0} with nontrivial finite-length intervals in general \citep{Lei2013distributionfree,Barber2020thelimits}. We therefore pursue a relaxation in the next section.

\subsection*{Notation}
Let $\bbP$ denote the probability law induced by $P$, and let $\bbE$ denote expectation under $\bbP$. We refer to the environment in which we predict $Y_{n+1}$ as the test environment, and to the corresponding probability law as the test law. Let $\mathrm{len}(C)$ be the length of a prediction interval $C \subseteq \bbR$.

\section{Target of Interest and IV Shift Relaxation}
This section defines the target prediction intervals. We begin with prediction intervals with exact conditional coverage, rewrite this target in moment form, and then relax it to coverage over a class of IV shifts. The resulting target applies to all radius classes considered in the paper.

\subsection{Exact Conditional Coverage}
For a generic interval $\widehat C_\tau$, exact conditional coverage in the IV setting takes the form
\begin{align}
\label{eq:exact_conditional}
\bbP\p{Y_{n+1}\in \widehat C_\tau(X_{n+1},Z_{n+1})\mid Z_{n+1}=z}=1-\alpha,\quad \forall z \in \calZ.
\end{align}
The next proposition rewrites \eqref{eq:exact_conditional} in moment form.

\begin{proposition}[Moment characterization of exact conditional coverage]\label{prop:moment_char}
Equation \eqref{eq:exact_conditional} holds if and only if
\begin{align}
\label{eq:moment_char_generic}
\bbE\sqb{f(Z_{n+1})\Bigp{\mathbbm{1}\sqb{Y_{n+1}\in \widehat C_\tau(X_{n+1},Z_{n+1})}-(1-\alpha)}}=0,
\quad \forall f\in\calM,
\end{align}
where $\calM$ denotes the set of all bounded measurable functions $f\colon \calZ\to\bbR$.
\end{proposition}

Proposition~\ref{prop:moment_char} shows that exact conditional coverage is equivalent to marginal coverage weighted by every measurable function of the IV. This equivalence also clarifies the impossibility result. Exact conditional coverage asks for validity after every possible reweighting of $Z$, which is too demanding in a distribution-free finite-sample setting. \citet{Lei2013distributionfree} shows that conditional guarantees of this type cannot generally be achieved by nontrivial finite-length prediction sets, and \citet{Barber2020thelimits} sharpens this message by showing that even approximate conditional coverage requires intervals that are essentially as conservative as those built for much larger subsets. The conditioning variable is the IV rather than a generic predictor, but the impossibility itself is unchanged.

\subsection{Relaxation via IV Shifts}
We therefore replace the full measurable class $\calM$ in \eqref{eq:moment_char_generic} with a user-specified class $\calF\subset \calM$. This yields the relaxed requirement
\begin{align}
\label{eq:moments_F}
\bbE\sqb{f(Z_{n+1})\Bigp{\mathbbm{1}\sqb{Y_{n+1}\in \widehat C_\tau(X_{n+1},Z_{n+1})}-(1-\alpha)}}\ge 0,
\quad \forall f\in\calF.
\end{align}
Following \citet{Gibbs2025conformalprediction}, \eqref{eq:moments_F} can be interpreted as robust marginal coverage over a family of IV shifts. For any nonnegative $f\in\calF$ with $\bbE\sqb{f(Z)}>0$, define a tilted distribution
\begin{align}
\label{eq:tilt_family}
\frac{d\bbP_f(y,x,z)}{d\bbP(y,x,z)} = \frac{f(z)}{\bbE\sqb{f(Z)}}.
\end{align}
Then \eqref{eq:moments_F} implies
\[
\bbP_f\p{Y_{n+1}\in \widehat C_\tau(X_{n+1},Z_{n+1})}\ge 1-\alpha
\]
for every such $f$. In applications with IVs, this interpretation is natural. Policies often change the marginal distribution of the IV while leaving the structural relationship \eqref{eq:npiv} intact. The one-sided inequality in \eqref{eq:moments_F} is enough for lower coverage guarantees, which is the finite-sample target throughout the paper.

\subsection{Prediction Intervals and Radius Classes}
In constructing prediction intervals $\widehat C_\tau(x,z)$, we allow the use of an estimator $\widehat h$ of the structural function $h_0$. Let $\widehat h\colon \calX\to\calY$ denote a predictor of $h_0$ obtained from some NPIV method. Then we define the general family of intervals as
\begin{align}
\label{eq:general_interval}
\widehat C_\tau(x,z)=\sqb{\widehat h(x)-\widehat\tau(x,z),\widehat h(x)+\widehat\tau(x,z)},
\end{align}
where $\widehat\tau\colon \calX\times\calZ\to\bbR_+$ is a nonnegative radius function.

Based on their dependence on the variables, we distinguish the following three radius classes:
\begin{align*}
\calT_{XZ}&\coloneqq \cb{\tau\colon \calX\times\calZ\to\bbR_+\text{ measurable}},\\
\calT_Z&\coloneqq \cb{\tau\colon \exists\tau_Z\colon \calZ\to\bbR_+\text{ such that }\tau(x,z)=\tau_Z(z)},\\
\calT_X&\coloneqq \cb{\tau\colon \exists\tau_X\colon \calX\to\bbR_+\text{ such that }\tau(x,z)=\tau_X(x)}.
\end{align*}
The broad class $\calT_{XZ}$ lets the radius vary jointly with $(X,Z)$ and serves as the common parent class for the later specializations. The IV-indexed class $\calT_Z$ keeps the center equal to $\widehat h(X)$ while allowing the radius to adapt to the IV. The $X$-indexed class $\calT_X$ yields a prediction interval of the form
\[
\widehat C_X(x)=\sqb{\widehat h(x)-\widehat\tau_X(x),\widehat h(x)+\widehat\tau_X(x)}.
\]
The common special case $\calT_X\cap\calT_Z$ is the constant-radius split conformal rule.

The distinction among the three classes is substantive. The class $\calT_{XZ}$ is the most expressive benchmark. The class $\calT_Z$ is IV-specific, because the IV indexes predictive uncertainty while the center remains a function of $X$ alone. The class $\calT_X$ yields prediction intervals $C(X)$ that depend only on $X$, which is the most natural class from the NPIV motivation. This paper studies all three classes, with the main methodological focus on $\calT_X$.

\subsection{A Radius-Class Oracle Problem}
The relaxation \eqref{eq:moments_F} is common across the three radius classes, but their feasible sets differ. Let $\calT\subset \calT_{XZ}$ be any radius class, and define the oracle problem
\begin{align}
\label{eq:oracle_radius_problem}
&\inf_{\tau\in\calT} \bbE\sqb{\tau(X,Z)}\\
\label{eq:oracle_radius_constraint}
&\text{s.t.}\ \ \bbE\sqb{f(Z)\Bigp{\mathbbm{1}\sqb{\abs{Y-h_0(X)}\le \tau(X,Z)}-(1-\alpha)}}\ge 0,
\quad \forall f\in\calF.
\end{align}
Let $L(\calT,\calF)$ denote the infimum in \eqref{eq:oracle_radius_problem} under the constraint \eqref{eq:oracle_radius_constraint}, that is,
$L(\calT,\calF)=\inf\cb{\bbE\sqb{\tau(X,Z)}:\tau\in\calT \text{ satisfies }\eqref{eq:oracle_radius_constraint}}$.
Here $L(\calT,\calF)$ is the oracle expected half-length within the radius class $\calT$ under the shift-robust target generated by $\calF$.

\begin{proposition}[Radius-class nesting]\label{prop:radius_nesting}
Suppose $\calT_1\subset \calT_2\subset \calT_{XZ}$ and that the feasible set defined by \eqref{eq:oracle_radius_problem} and \eqref{eq:oracle_radius_constraint} is nonempty. Then we have
\[
L(\calT_2,\calF)\le L(\calT_1,\calF).
\]
In particular, we have
\[
L(\calT_{XZ},\calF)\le L(\calT_Z,\calF),\qquad L(\calT_{XZ},\calF)\le L(\calT_X,\calF).
\]
\end{proposition}

Proposition~\ref{prop:radius_nesting} clarifies the trade-off. The class $\calT_{XZ}$ is the broadest feasible class and therefore the easiest one in which to reduce oracle length. The classes $\calT_Z$ and $\calT_X$ are more restrictive, but they have sharper substantive interpretations. The class $\calT_Z$ treats the IV as an index of residual uncertainty. The class $\calT_X$ yields prediction intervals $C(X)$ depending only on $X$, which is the most natural class from the NPIV motivation.

\subsection{Connection to Minimax and Importance Weighting}
The constrained problem defined by \eqref{eq:oracle_radius_problem} and \eqref{eq:oracle_radius_constraint} has the same minimax flavor as conditional-moment formulations in NPIV. In the continuum-of-moments and linear-inverse viewpoints of \citet{Carrasco2000generalizationof} and \citet{Carrasco2007linearinverse}, conditional moments are converted into a family of unconditional moments indexed by a function class. More recent adversarial procedures, such as DeepGMM and the minimax-optimal estimators of \citet{Bennett2019deepgeneralized} and \citet{Dikkala2020minimaxestimation}, estimate $h_0$ by searching for moment violations over a critic class.

Our use of the class $\calF$ is different. We are not trying to identify $h_0$ by finding moments that are most informative about the structural equation. Instead, $h_0$ or its estimator $\widehat h$ is treated as fixed, and the radius $\tau$ is chosen so that the coverage residual $\mathbbm{1}\sqb{Y\in \widehat C_\tau(X,Z)}-(1-\alpha)$ has nonnegative moments against every $f\in\calF$. Thus, $\calF$ is not a critic class for structural estimation.

This viewpoint also connects naturally to importance weighting. Equation \eqref{eq:tilt_family} rewrites the relaxed coverage target as a family of ordinary marginal coverage statements under weighted laws. This is the same algebraic mechanism that underlies covariate-shift correction in \citet{Shimodaira2000improvingpredictive} and importance-weighting formulations for NPIV such as \citet{Kato2022learningcausal}. In our setting, however, the weights define the probability laws over which coverage is required.

\section{IV-Conditional Conformal Prediction}
In this section, we propose IV-conditional conformal prediction (IV-CCP). The details differ across the three target radius classes, $\calT_{XZ}$, $\calT_Z$, and $\calT_X$.

\subsection{Split Conformal Prediction for NPIV}
In all cases, we construct prediction intervals using a split-type implementation of conformal prediction. We first split the sample into a training set $\mathcal{I}_{\mathrm{tr}}$ and a calibration set $\mathcal{I}_{\mathrm{cal}}$ with $|\mathcal{I}_{\mathrm{cal}}|=m$. We then fit any NPIV estimator $\widehat h$ on $\mathcal{I}_{\mathrm{tr}}$. On the calibration set, we compute the following nonconformity scores:
\begin{align}
\label{eq:scores}
S_i=\abs{Y_i-\widehat h(X_i)},\qquad i\in\mathcal{I}_{\mathrm{cal}}.
\end{align}
Conditional on the training split, $\widehat h$ is fixed, so the calibration scores and the test score are exchangeable.

\subsection{IV-CCP with \texorpdfstring{$(X, Z)$}{XZ}-Indexed Radius}
The first method studies the broad benchmark class $\calT_{XZ}$. This class is not the main target of the paper, but it is mathematically convenient because the same variable $(X,Z)$ indexes both the radius and the shift family.

\paragraph{Class of the joint shift.}
Let $W=(X,Z)$ and let $\phi_W\colon \calX\times\calZ\to\bbR^d$ be a feature map. Define
\[
\calG_W=\cb{w\mapsto \beta^\top\phi_W(w): \beta\in\bbR^d},\qquad
\calF_W=\calG_W\cap\cb{f\ge 0}.
\]
Throughout this section, we take the first component of $\phi_W$ to be the constant $1$, so $\calG_W$ contains the constant functions. The class $\calG_W$ determines how the radius may vary with $(X,Z)$, and $\calF_W$ determines the family of joint distribution shifts.

\paragraph{Calibration.}
Fix a test pair $w=(x,z)$ and a candidate radius $s\ge 0$. Denote the pinball loss by
\[\rho_\alpha(u)=\alpha\max\cb{u,0}+(1-\alpha)\max\cb{-u,0}.\]
Using the calibration pairs $(W_i,S_i)$ and the augmented point $(w,s)$, consider
\begin{align*}
\widehat g_s\in\argmin_{g\in\calG_W}
\frac{1}{m+1}\sum_{i\in\mathcal{I}_{\mathrm{cal}}}\rho_\alpha\p{S_i-g(W_i)}+
\frac{1}{m+1}\rho_\alpha\p{s-g(w)}.
\end{align*}
This is the finite-dimensional conditional conformal machinery of \citet{Gibbs2025conformalprediction} with covariate $W=(X,Z)$. Form the LP dual corresponding to this augmented quantile regression problem, and let $\eta_{m+1}(s)$ denote the dual coordinate associated with the augmented point. As in \citet{Gibbs2025conformalprediction}, $\eta_{m+1}(s)$ is nondecreasing in $s$, and $s$ belongs to the conformal upper set if and only if $\eta_{m+1}(s)<1-\alpha$. We therefore define $\widehat\tau_{XZ}(x,z)$ as the largest accepted value of $s$ and output
\begin{align}
\label{eq:joint_interval}
\widehat C_{XZ}(x,z)=\sqb{\widehat h(x)-\widehat\tau_{XZ}(x,z),\widehat h(x)+\widehat\tau_{XZ}(x,z)}.
\end{align}
For any nonnegative $f\in\calF_W$ with $\bbE\sqb{f(W)}>0$, define the joint tilted law
\begin{align}
\label{eq:joint_tilt}
\frac{d\bbP^{W}_f(y,x,z)}{d\bbP(y,x,z)}=\frac{f(x,z)}{\bbE\sqb{f(W)}}.
\end{align}
The finite-sample guarantee for this class is stated in Theorem~\ref{thm:joint_shift_robust}.

\subsection{IV-CCP with \texorpdfstring{$Z$}{Z}-Indexed Radius}
The second method studies the IV-indexed class $\calT_Z$. This is the exact IV-specific conformal class, and it is the first main focus of the paper. The interval center remains $\widehat h(X)$, while the radius depends on $Z$ alone.

\paragraph{Class of the IV shift.}
Let $\phi\colon \calZ\to\bbR^d$ be a feature map, and define
\[
\calG=\cb{z'\mapsto \beta^\top\phi(z'): \beta\in\bbR^d},\qquad
\calF=\calG\cap\cb{f\ge 0}.
\]
Throughout the sequel, we take the first component of $\phi$ to be the constant $1$, so $\calG$ contains the constant functions. Fix a test IV $z$ and a candidate radius $s\ge 0$.

\paragraph{Calibration.}
Using the same scores \eqref{eq:scores}, consider
\begin{align*}
\widehat g_s\in\argmin_{g\in\calG}
\frac{1}{m+1}\sum_{i\in\mathcal{I}_{\mathrm{cal}}}\rho_\alpha\p{S_i-g(Z_i)}+
\frac{1}{m+1}\rho_\alpha\p{s-g(z)}.
\end{align*}
This is the augmented quantile-regression problem of \citet{Gibbs2025conformalprediction} with covariate $Z$. Let $\eta_{m+1}(s)$ denote the dual variable associated with the augmented point. As before, $\eta_{m+1}(s)$ is nondecreasing in $s$, and the largest accepted value of $s$ defines the radius $\widehat\tau_Z(z)$. The resulting IV-indexed interval is
\begin{align}
\label{eq:iv_indexed_interval}
\widehat C_Z(x,z)=\sqb{\widehat h(x)-\widehat\tau_Z(z),\widehat h(x)+\widehat\tau_Z(z)}.
\end{align}

\begin{algorithm}[t]
\caption{IV-CCP interval at a test point $(x,z)$}
\label{alg:ivccp}
\begin{algorithmic}[1]
\STATE Input: training data $\mathcal{I}_{\mathrm{tr}}$, calibration data $\mathcal{I}_{\mathrm{cal}}$, level $\alpha$, features $\phi$
\STATE Fit NPIV on $\mathcal{I}_{\mathrm{tr}}$ to get $\widehat h$
\STATE Compute calibration scores $S_i=\abs{Y_i-\widehat h(X_i)}$ for $i\in\mathcal{I}_{\mathrm{cal}}$
\STATE For a test point $z$, compute $\widehat\tau_Z(z)$ via binary search over $s\ge 0$ using the LP dual membership test
\STATE Return: $\widehat C_Z(x,z)=\sqb{\widehat h(x)-\widehat\tau_Z(z),\widehat h(x)+\widehat\tau_Z(z)}$
\end{algorithmic}
\end{algorithm}

\paragraph{Interpretation}
The class $\calT_Z$ is less flexible than $\calT_{XZ}$, but it restores the IV-specific interpretation that the relevant change occurs in the IV distribution. The interval center remains a function of $X$ alone. What changes with $Z$ is the residual uncertainty around that center. This is compatible with the exclusion restriction in \eqref{eq:npiv}, because the IV does not enter the structural center directly.

\paragraph{Relationship to weighted conformal prediction.}
Weighted conformal prediction under covariate shift \citep{Tibshirani2019conformalprediction} protects against one fixed target distribution, provided the corresponding density ratio is known or can be estimated. IV-CCP is different. It protects simultaneously against every IV tilt in the class $\calF$. In other words, it is a guarantee over a family of shifts rather than a guarantee for a single shift. This distinction is central in IV applications, where the relevant uncertainty is often not one fully specified probability law in a setting of interest but a set of plausible changes in the IV distribution.

\subsection{IV-CCP with \texorpdfstring{$X$}{X}-Indexed Radius}
The third method studies the class $\calT_X$, which produces prediction intervals $C(X)$ depending only on $X$. This is the most natural target in NPIV, because the final prediction rule depends on the structural regressor alone. It is also the second main focus of the paper. The difficulty is that the interval is indexed by $X$, while the robustness target is indexed by shifts in $Z$.

\paragraph{IV-CCP for the $X$-indexed interval}
In this case, we consider a prediction interval of the following form:
\begin{align}
\label{eq:X_indexed_interval}
\widehat C_X(x)=\sqb{\widehat h(x)-\widehat\tau_X(x),\widehat h(x)+\widehat\tau_X(x)}.
\end{align}
Let us define the nonconformity score as
\[
S=\abs{Y-\widehat h(X)}.
\]
For a candidate radius function $\tau_X\colon \calX\to\bbR_+$, define the coverage residual
\begin{align}
\label{eq:xindexed_moment_function}
m_{\tau_X}(S,X)=\mathbbm{1}\sqb{S\le \tau_X(X)}-(1-\alpha).
\end{align}
Then exact conditional coverage for $\widehat C_X$ is equivalent to
\begin{align}
\label{eq:xindexed_conditional_moment}
\bbE\sqb{m_{\tau_X}(S,X)\mid Z=z}=0,\quad \forall z.
\end{align}
The relaxed family-of-shifts target becomes
\begin{align}
\label{eq:xindexed_F_moment}
\bbE\sqb{f(Z)m_{\tau_X}(S,X)}\ge 0,\quad \forall f\in\calF.
\end{align}
The challenge is to construct an $X$-indexed radius while the protected shift family remains a class of functions of $Z$.

\paragraph{Importance-weighted representation}
The importance-weighting approach of \citet{Kato2022learningcausal} provides a natural route. Suppose the conditional density ratio
\begin{align}
\label{eq:ratio_sx_z}
r_0(s,x\mid z)=\frac{p_{S,X\mid Z}(s,x\mid z)}{p_{S,X}(s,x)}
\end{align}
is well defined. Then conditional moments of $(S,X)$ given $Z$ can be rewritten as weighted unconditional moments.

\begin{proposition}[Importance-weighted identity for the $X$-indexed class]\label{prop:iw_identity}
Suppose the conditional density ratio \eqref{eq:ratio_sx_z} exists. Then for any measurable $\tau_X\colon \calX\to\bbR_+$,
\begin{align}
\label{eq:iw_identity}
\bbE\sqb{m_{\tau_X}(S,X)\mid Z=z}=\bbE\sqb{m_{\tau_X}(S,X)r_0(S,X\mid z)},\quad \forall z.
\end{align}
Consequently, \eqref{eq:xindexed_conditional_moment} is equivalent to
\begin{align}
\label{eq:iw_exact_target}
\bbE\sqb{m_{\tau_X}(S,X)r_0(S,X\mid z)}=0,\quad \forall z.
\end{align}
\end{proposition}

Proposition~\ref{prop:iw_identity} is the key bridge from the conditional coverage target to unconditional weighted moments. It is exactly the same algebraic step that \citet{Kato2022learningcausal} uses for general conditional moment learning, now applied to the interval residual $m_{\tau_X}(S,X)$.

To turn this identity into a learning criterion, let $\calR_X\subset \cb{\tau_X\colon \calX\to\bbR_+}$ denote a chosen model class for the radius function, for example a sieve span, a neural network class, or a shape-restricted class. Let $\widehat r(s,x\mid z)$ be an estimated version of \eqref{eq:ratio_sx_z}, obtained on the training split or by cross-fitting. Let $\widetilde Z_1,\dots,\widetilde Z_M$ be evaluation points, for example the calibration IVs themselves or a deterministic grid. Because the indicator in \eqref{eq:xindexed_moment_function} is discontinuous, it is convenient in optimization to replace it with a smooth surrogate. Let $\psi_\kappa\colon \bbR\to\bbR$ be a smooth approximation to $u\mapsto \mathbbm{1}\sqb{u\ge 0}-(1-\alpha)$, indexed by a smoothing parameter $\kappa>0$.

We then consider the criterion
\begin{align}
\label{eq:iw_objective_sample}
\widehat Q_n(\tau_X)=\frac{1}{M}\sum_{j=1}^M\Bigg(\frac{1}{m}\sum_{i\in\mathcal I_{\mathrm{cal}}}\psi_\kappa\p{\tau_X(X_i)-S_i}\widehat r\p{S_i,X_i\mid \widetilde Z_j}\Bigg)^2
\end{align}
and the estimator
\begin{align}
\label{eq:iw_tau_estimator}
\widehat\tau_X\in\argmin_{\tau_X\in\calR_X}\cb{\frac{1}{m}\sum_{i\in\mathcal I_{\mathrm{cal}}}\tau_X(X_i)+\lambda\widehat Q_n(\tau_X)},
\end{align}
where $\lambda>0$ balances average interval length and weighted moment fit. The resulting prediction interval is \eqref{eq:X_indexed_interval}. This procedure is not distribution-free in finite samples, because it relies on density-ratio estimation and surrogate optimization. Its role is different. It provides a principled way to construct an $X$-indexed radius under the same conditional-moment logic that underlies NPIV.

\paragraph{Single-shift weighted conformal recalibration}
The importance-weighted procedure above learns the shape of the $X$-indexed radius. To obtain a finite-sample guarantee for a fixed target distribution shift while keeping the final interval $X$-indexed, one can add a final scalar conformal recalibration step in the spirit of \citet{Tibshirani2019conformalprediction}. Let $q\colon \calX\to\bbR_+$ be a positive function learned on a sample independent of the final recalibration split, for example by \eqref{eq:iw_tau_estimator}, or by any other method. Let $\mathcal I_{\mathrm{rcal}}$ denote the final recalibration split. Consider intervals of the form
\begin{align}
\label{eq:xindexed_scaled_interval}
\widehat C_{X,f_0}(x)=\sqb{\widehat h(x)-\widehat t_{f_0}q(x),\widehat h(x)+\widehat t_{f_0}q(x)},
\end{align}
where $f_0\in\calF$ is a fixed target shift. Form the normalized scores as
\begin{align}
\label{eq:normalized_scores_tx}
R_i=\frac{\abs{Y_i-\widehat h(X_i)}}{q(X_i)},\qquad i\in\mathcal I_{\mathrm{rcal}}.
\end{align}
The weights are the density ratios induced by the target shift, which are defined as
\begin{align}
\label{eq:single_shift_weights}
w_{f_0}(z)=\frac{f_0(z)}{\bbE\sqb{f_0(Z)}}.
\end{align}
Assume furthermore that there is a known finite constant $B_{f_0}$ such that
\[
w_{f_0}(z)\le B_{f_0}\qquad \text{for all $z$ in the support of $Z$.}
\]
Define $\widehat t_{f_0}$ as the conservative weighted split conformal cutoff obtained from the calibration weights $w_{f_0}(Z_i)$ and the worst-case test weight $B_{f_0}$, that is,
\[
\widehat t_{f_0}=\inf\cb{t\in\bbR_+: \sum_{i\in\mathcal I_{\mathrm{rcal}}}\frac{w_{f_0}(Z_i)}{\sum_{j\in\mathcal I_{\mathrm{rcal}}}w_{f_0}(Z_j)+B_{f_0}}\mathbbm{1}\sqb{R_i\le t}\ge 1-\alpha},
\]
with the convention $\inf\emptyset=\infty$. By construction, the cutoff is scalar and the final interval remains a function of $X$ only. Because it upper bounds the test-point-specific weighted split conformal cutoff of \citet{Tibshirani2019conformalprediction}, it yields the following finite-sample guarantee under the single target law $\bbP_{f_0}$:
\[
\bbP_{f_0}\p{Y_{n+1}\in \widehat C_{X,f_0}(X_{n+1})}\ge 1-\alpha,
\]
where the detailed statement is shown in Theorem~\ref{thm:weighted_conformal_tx}. 

This coverage guarantee is intentionally narrower than the guarantees for $\calT_{XZ}$ and $\calT_Z$. It holds for one specified distribution shift rather than uniformly over a whole class. This reflects the current methodological gap. Family-wise exact finite-sample coverage for the mixed-index target $C(X)$ is not presently available, whereas a conservative single-shift weighted conformal recalibration is available.

\subsection{Discussion}
The three classes now fall into place. The benchmark class $\calT_{XZ}$ is the broad exact conformal class. The class $\calT_Z$ is the exact IV-specific class that delivers simultaneous finite-sample coverage over a family of IV tilts. The class $\calT_X$ is the primary target of interest, but it currently requires a different strategy, namely, importance-weighted conditional-moment learning and, when needed, a final single-shift weighted conformal recalibration.

\section{Designing the IV Shift Class \texorpdfstring{$\calF$}{F}}
The choice of $\calF$ is the main modeling decision in the paper. Through \eqref{eq:tilt_family}, it determines the ambiguity set of shifted distributions under which coverage is required. Accordingly, $\calF$ should encode how the marginal distribution of the IV may change, not how $Y$ depends on $Z$, and not how $\widehat h$ is estimated.

\subsection{From Feature Maps to Ambiguity Sets}
For the exact classes $\calT_{XZ}$ and $\calT_Z$, the shift class is induced by a finite-dimensional feature map. For $\calT_Z$, with feature map $\phi\colon \calZ\to\bbR^d$, define
\[
\calG_\phi=\cb{z\mapsto \beta^\top\phi(z):\beta\in\bbR^d},\qquad
\calF_\phi=\calG_\phi\cap\cb{f\ge 0}.
\]
The corresponding ambiguity set is
\[
\calQ(\calF_\phi)=\cb{\bbP_f: \frac{d\bbP_f(y,x,z)}{d\bbP(y,x,z)}=\frac{f(z)}{\bbE\sqb{f(Z)}},\ f\in\calF_\phi,\ \bbE\sqb{f(Z)}>0}.
\]
If $\calF_1\subset \calF_2$, then $\calQ(\calF_1)\subset \calQ(\calF_2)$. Enlarging $\calF$ therefore protects against a broader family of shifted distributions. If the constant function belongs to $\calF$, then the training distribution itself belongs to the ambiguity set. In practice, this is enforced by taking the first component of $\phi$ to be $1$.

The gain in robustness is not free. In the exact IV-indexed class, Theorem~\ref{thm:length} shows that the calibration contribution to interval length scales with $d/(m+1)$. Thus, the design problem is to choose the smallest class $\calF$ whose ambiguity set still contains the distribution shifts that matter scientifically or empirically.

\subsection{Finite-Dimensional Templates for Exact Conformal Classes}
\paragraph{Indicator and partition classes}
If $Z$ is discrete, or if the user is willing to discretize a continuous IV, the simplest construction is an indicator basis. For a partition $\Pi=\cb{\calZ_1,\dots,\calZ_G}$ of the IV space, take
\[
\phi(z)=\p{1,\mathbbm{1}\sqb{z\in \calZ_1},\dots,\mathbbm{1}\sqb{z\in \calZ_G}}.
\]
Then the induced tilts are constant within cells, and the fitted radius is piecewise constant. This is the natural construction for binary encouragement designs, judge identifiers, hospital identifiers, or other categorical IVs.

The same idea extends to overlapping groups. If $G_1,\dots,G_L\subset \calZ$ are policy-relevant subsets, one may take
\[
\phi(z)=\p{1,\mathbbm{1}\sqb{z\in G_1},\dots,\mathbbm{1}\sqb{z\in G_L}}.
\]
The resulting guarantee holds for every group indicator and every nonnegative linear combination of them. When the groups overlap, the fitted radius is constant on the atoms generated by the overlap pattern.

\paragraph{Multiscale discretization}
For a continuous scalar IV, one may include indicators for a coarse partition together with indicators for a finer nested partition. This allows the ambiguity set to contain both global reweightings and more local perturbations of the IV law. When $Z$ is irregular or mixed discrete-continuous, tree-based or forest-based partitions play the same role. A tree learned on $\mathcal I_{\mathrm{tr}}$ induces leaves $L_1,\dots,L_G$, and one can use leaf indicators $\mathbbm{1}\sqb{z\in L_g}$ as features. This yields an adaptive discretization of the IV space.

\paragraph{Series and sieve classes}
When $Z$ is continuous and low-dimensional, a natural alternative is a smooth series basis defined as
\[
\phi(z)=\p{1,\psi_1(z),\dots,\psi_d(z)},
\]
where $\psi_1,\dots,\psi_d$ may be spline basis functions, orthogonal polynomials, wavelets, or trigonometric terms. The resulting class protects against smooth density-ratio perturbations that can be represented, or well approximated, by the sieve span. This is the same approximation logic that underlies sieve NPIV estimation \citep{Chen2012estimationof,Chen2015sievewald}.

For a vector IV $Z=(Z_1,\dots,Z_{k_Z})$, one may use additive or tensor-product sieves. An additive construction takes the form
\[
\phi(z)=\Bigp{1,\psi_{1,1}(z_1),\dots,\psi_{1,d_1}(z_1),\dots,\psi_{k_Z,1}(z_{k_Z}),\dots,\psi_{k_Z,d_{k_Z}}(z_{k_Z})},
\]
while a tensor-product construction adds selected interaction terms of the form
\[
\phi_{j_1,\dots,j_q}(z)=\prod_{\ell=1}^q \psi_{\ell,j_\ell}(z_\ell).
\]
The additive version is often preferable when $k_Z$ is moderate, because it keeps $d$ small. Tensor products are appropriate only when domain knowledge points to specific interactions that may change in the setting of interest.

\begin{remark}
Not every useful class requires a large basis. In many IV applications, a low-dimensional semiparametric summary is enough. Examples include
\[
\phi(z)=\p{1,z},\qquad \phi(z)=\p{1,z,z^2},\qquad \phi(z)=\p{1,\text{judge indicators},\text{judge-by-time interactions}}.
\]
These classes are especially attractive when policy discussions already single out the relevant directions of change, for example a shift in encouragement probability, a drift in judge shares, or a smooth increase in the mass of distant households.
\end{remark}

\subsection{Richer Classes through Finite-Dimensional Approximation}
\paragraph{RKHS and kernel dictionaries}
A convenient idealization for smooth, local, and nonlinear shifts is an RKHS ball. Let $K\colon \calZ\times\calZ\to\bbR$ be a positive-definite kernel with RKHS $\calH_K$. A natural infinite-dimensional candidate is
\[
\calF^{\infty}_{K,B}=\cb{f\in \calH_K: \norm{f}_{\calH_K}\le B,\ f\ge 0}.
\]
The exact finite-sample result in this paper is developed for finite-dimensional classes, so the practical route is to approximate $\calH_K$ by a finite dictionary,
\[
\phi_r(z)=\p{1,K(z,u_1),\dots,K(z,u_r)},
\]
where $u_1,\dots,u_r\in \calZ$ are landmark points. The induced class $\calF_{\phi_r}$ is then a finite-dimensional proxy for the ideal RKHS family, and the exact theorem applies to this proxy.

The kernel determines the geometry of the protected shifts. Gaussian kernels generate local tilts centered around the landmarks. Polynomial kernels generate global low-order smooth reweightings. If one expects the shift to be smooth after an appropriate transformation of the IV, one may replace $z$ by a representation $g(z)$ and use features of the form $K\p{g(z),g(u_j)}$. Low-rank approximations such as Nystr\"om dictionaries and random features provide practical ways to keep $d$ moderate \citep{Rahimi2007randomfeatures,Musco2017recursivesampling}.

\paragraph{Sparse linear and representation-based classes}
When the IV is high-dimensional, a direct sieve or kernel expansion in the raw coordinates can make $d$ too large relative to the calibration size. A natural compromise is to construct a low-dimensional summary map $\psi\colon \calZ\to\bbR^q$ on the training split and then set
\[
\phi(z)=\p{1,\psi_1(z),\dots,\psi_q(z)}.
\]
The summary may consist of principal components, selected coordinates, geographic summaries, judge-share summaries, or a learned representation of a text, image, or network-valued IV. The point is not that the representation must be rich enough to estimate $h_0$. It only needs to span the directions along which the IV distribution is likely to move in the setting of interest.

\paragraph{Shape-restricted classes}
For scalar IVs, or for low-dimensional summaries of vector IVs, it can be natural to focus on monotone, convex, or bounded-variation shifts. Examples include interventions that only increase encouragement among larger values of $Z$, or geographic changes that monotonically reweight distance. In the exact conformal setting, the cleanest implementation is to approximate a shape-restricted family by a finite grid and a basis of piecewise-constant or piecewise-linear functions. One may then use the induced low-dimensional span as a proxy for the desired cone. A stronger approach would impose linear inequality constraints on the coefficients. Such constrained calibration remains computationally natural, because it leads to linear or convex programs, but it goes beyond the unconstrained span class analyzed here.

\subsection{The Additional Modeling Layer in the \texorpdfstring{$X$}{X}-Indexed Class}
For the $X$-indexed class $\calT_X$, there are two distinct modeling choices. The first is the shift class $\calF$, which still describes how the IV distribution may move in the setting of interest. The second is the radius class for $\tau_X$, which determines how the interval width may vary with $X$. The two should not be conflated.

In practice, one may choose the radius class on $X$ from the same broad families used in nonparametric regression, for example bins in $X$, additive sieves, tensor-product bases, shape-restricted classes, RKHS models, or neural-network classes. The role of that class is completely different from the role of $\calF$. The radius class controls how much heterogeneity the final prediction interval may have across values of $X$. The shift class controls which settings of interest must be protected. In the exact classes $\calT_{XZ}$ and $\calT_Z$, these two roles are partially tied together by the conditional conformal algorithm. In the $X$-indexed method, they are separate.

\section{Theory}
This section collects the main theoretical statements. Proofs appear in Appendix~\ref{appdx:proofs}.

\subsection{Exact Finite-Sample Coverage for \texorpdfstring{$(X, Z)$}{XZ}-Indexed Radius}
We first show that the coverage ratio with an $(X, Z)$-indexed radius works as intended.

\begin{theorem}[Shift-robust coverage over joint tilts]\label{thm:joint_shift_robust}
Assume that the calibration sample and the test point are exchangeable conditional on the training split. Let $\widehat C_{XZ}(x,z)$ be produced by the augmented quantile-regression calibration of Section 4 using feature map $\phi_W$ and the associated nonnegative linear shift class $\calF_W$. Then for any measurable $f\in\calF_W$ with $f\ge 0$ and $\bbE\sqb{f(W)}>0$, the interval satisfies
\[
\bbP^W_f\p{Y_{n+1}\in \widehat C_{XZ}(X_{n+1},Z_{n+1})}\ge 1-\alpha,
\]
where $\bbP^W_f$ is the tilted distribution in \eqref{eq:joint_tilt}.
\end{theorem}

\subsection{Exact Finite-Sample Coverage for \texorpdfstring{$Z$}{Z}-Indexed Radius}
We next show that the coverage ratio with a $Z$-indexed radius works as intended.

\begin{theorem}[Shift-robust coverage over IV tilts]\label{thm:shift_robust}
Assume that the calibration sample and the test point are exchangeable conditional on the training split. Let $\widehat C_Z(x,z)$ be produced by Algorithm~\ref{alg:ivccp} using feature map $\phi$ and the associated nonnegative linear shift class $\calF$. Then for any measurable $f\in\calF$ with $f\ge 0$ and $\bbE\sqb{f(Z)}>0$, the interval satisfies
\[
\bbP_f\p{Y_{n+1}\in \widehat C_Z(X_{n+1},Z_{n+1})}\ge 1-\alpha,
\]
where $\bbP_f$ is the tilted distribution in \eqref{eq:tilt_family}.
\end{theorem}

\begin{remark}
Theorem~\ref{thm:shift_robust} is distribution-free, because it does not require \eqref{eq:npiv} to hold. The class $\calT_X$ remains the main target of interest, but the exact finite-sample guarantee currently available under IV shifts is the IV-indexed class $\calT_Z$.
\end{remark}

\paragraph{Oracle comparison for the IV-indexed class}
We now relate interval length in the exact IV-indexed class to NPIV estimation error and calibration complexity.

Let the oracle structural function be $h_0$ and define oracle scores $S_i^\star=\abs{Y_i-h_0(X_i)}$. Let $\widehat\tau_Z$ be computed from $\cb{S_i}$ and let $\widehat\tau_Z^\star$ be computed from $\cb{S_i^\star}$ using the same IV-indexed calibration algorithm and features.

\begin{proposition}[Stability to NPIV estimation error]\label{prop:stability}
For any realization of the data and any $z\in\calZ$, we have
\[
\abs{\widehat\tau_Z(z)-\widehat\tau_Z^\star(z)} \le \norm{\widehat h-h_0}_\infty.
\]
Consequently, we have
\[
\abs{\mathrm{len}(\widehat C_Z(x,z)) - \mathrm{len}(\widehat C_Z^\star(x,z))}\le 2\norm{\widehat h-h_0}_\infty,
\]
where $\widehat C_Z^\star(x,z)=\sqb{\widehat h(x)-\widehat\tau_Z^\star(z),\widehat h(x)+\widehat\tau_Z^\star(z)}$.
\end{proposition}

Next, compare this with the oracle IV-indexed radius
\[
\tau_0(z)=\inf\cb{t: \bbP\p{\abs{\varepsilon}\le t\mid Z=z}\ge 1-\alpha}.
\]
We require two regularity conditions.

\begin{assumption}[Local density lower bound]\label{ass:density}
For all $z\in\calZ$, the conditional distribution of $\abs{\varepsilon}$ given $Z=z$ has a Lebesgue density $p_{\abs{\varepsilon}\mid z}(\cdot)$ that is bounded below by $p_{\min}>0$ on a neighborhood of $\tau_0(z)$.
\end{assumption}

\begin{assumption}[Oracle quantile-transfer regularity]\label{ass:transfer}
Assume that the oracle radius function $z'\mapsto \tau_0(z')$ belongs to $\calG=\cb{z'\mapsto \beta^\top \phi(z'): \beta\in\bbR^d}$. For each fixed $z\in\calZ$, let $\widehat g_z\in\calG$ be an optimal basic solution of the augmented oracle quantile-regression problem defined in the proof of Theorem~\ref{thm:length}, chosen so that $\widehat g_z(z)=\widehat\tau_Z^\star(z)$. Then
\[
\abs{F_z\p{\widehat\tau_Z^\star(z)}-\frac{1}{m+1}\sum_{i=1}^{m+1}\mathbbm{1}\sqb{S_i^\star\le \widehat g_z(Z_i)}}=o_p(1),
\]
where $F_z(t)=\bbP\p{\abs{\varepsilon}\le t\mid Z=z}$.
\end{assumption}

\begin{theorem}[Length inflation bound]\label{thm:length}
Assume \eqref{eq:npiv} and Assumptions~\ref{ass:density} and \ref{ass:transfer}. Then, for each fixed $z\in\calZ$,
\[
\abs{\widehat\tau_Z^\star(z)-\tau_0(z)}\le \frac{d}{(m+1)p_{\min}}+o_p(1).
\]
Combining this with Proposition~\ref{prop:stability}, we obtain
\[
\abs{\widehat\tau_Z(z)-\tau_0(z)}\le \norm{\widehat h-h_0}_\infty + \frac{d}{(m+1)p_{\min}}+o_p(1).
\]
In particular,
\[
\mathrm{len}(\widehat C_Z(x,z))-\mathrm{len}(C_0(x,z))
\le 2\norm{\widehat h-h_0}_\infty + \frac{2d}{(m+1)p_{\min}}+o_p(1),
\]
where $C_0(x,z)=\sqb{h_0(x)-\tau_0(z), h_0(x)+\tau_0(z)}$.
\end{theorem}

\begin{remark}
If $\tau_0$ does not lie in $\calG$, the same argument yields an additional approximation term $\inf_{g\in\calG}\abs{g(z)-\tau_0(z)}$.
\end{remark}

\subsection{Exact Finite-Sample Coverage for \texorpdfstring{$X$}{X}-Indexed Radius}
The next proposition records the population criterion underlying the $X$-indexed importance-weighted method.

\begin{proposition}[Integrated weighted-moment characterization]\label{prop:iw_characterization}
Suppose that the conditional density ratio \eqref{eq:ratio_sx_z} exists. Let $\nu$ be a measure on $\calZ$ with full support and define
\[
Q(\tau_X)=\int \Bigp{\bbE\sqb{m_{\tau_X}(S,X)r_0(S,X\mid z)}}^2\,\nu(dz).
\]
Then $Q(\tau_X)=0$ if and only if \eqref{eq:xindexed_conditional_moment} holds for $\nu$-almost every $z$. In particular, $Q(\tau_X)=0$ implies exact conditional coverage for $\widehat C_X$ whenever $\nu$ has full support and the conditional moment is continuous in $z$.
\end{proposition}

\begin{remark}
Proposition~\ref{prop:iw_characterization} explains why \eqref{eq:iw_objective_sample} is a natural empirical analogue. The criterion is not distribution-free finite-sample, because it depends on the estimated ratio $\widehat r$ and the surrogate $\psi_\kappa$, but it targets the correct population condition.
\end{remark}

Then, the following theorem provides a complementary finite-sample statement for the $X$-indexed class. It applies to one fixed target shift rather than to a whole class, and it requires an independent final recalibration split together with a known upper bound on the target weight. This is the exact finite-sample guarantee currently available for prediction intervals $C(X)$ under a fixed target shift.

\begin{theorem}[Single-shift weighted conformal recalibration for $C(X)$]\label{thm:weighted_conformal_tx}
Fix a nonnegative $f_0\in\calF$ with $\bbE\sqb{f_0(Z)}>0$. Suppose $\widehat h$ and $q$ are learned on samples independent of the final recalibration split $\mathcal I_{\mathrm{rcal}}$, the weights \eqref{eq:single_shift_weights} are known, and there exists a known finite constant $B_{f_0}$ satisfying $\sup_z w_{f_0}(z)\le B_{f_0}$. Let $\widehat t_{f_0}$ be defined as above from the normalized scores \eqref{eq:normalized_scores_tx}. Then the interval \eqref{eq:xindexed_scaled_interval} satisfies
\[
\bbP_{f_0}\p{Y_{n+1}\in \widehat C_{X,f_0}(X_{n+1})}\ge 1-\alpha.
\]
\end{theorem}

\section{Simulation}

\begin{table}[!th]
\caption{Results of the experiment using Dataset~1.}
\label{tab:dataset1}
    \centering
\resizebox{0.7\linewidth}{!}{
\begin{tabular}{llllrrrr}
\toprule
\multirow[c]{2}{*}{Radius} & \multirow[c]{2}{*}{IV Shift} & \multirow[c]{2}{*}{Radius Model} & \multirow[c]{2}{*}{Test Dist. of IVs} & \multicolumn{2}{c}{Cov. Ratio} & \multicolumn{2}{c}{Interval Length} \\
 &  & & &  Mean & SD & Mean & SD \\
\midrule
\multirow[t]{12}{*}{XZ-indexed} & \multirow[t]{4}{*}{Bins} & \multirow[t]{4}{*}{} & Observed & 0.902 & 0.020 & 4.456 & 9.446 \\
 &  &  & Linear Tilt & 0.901 & 0.024 & 4.288 & 7.181 \\
 &  &  & Local Tilt & 0.900 & 0.027 & 4.335 & 6.623 \\
 &  &  & Step Tilt & 0.901 & 0.022 & 4.329 & 7.742 \\
\cline{2-8} \cline{3-8}
 & \multirow[t]{4}{*}{Linear} & \multirow[t]{4}{*}{} & Observed & 0.908 & 0.021 & 4.080 & 7.179 \\
 &  &  & Linear Tilt & 0.907 & 0.028 & 4.082 & 6.576 \\
 &  &  & Local Tilt & 0.906 & 0.032 & 4.076 & 6.175 \\
 &  &  & Step Tilt & 0.907 & 0.026 & 4.065 & 6.564 \\
\cline{2-8} \cline{3-8}
 & \multirow[t]{4}{*}{RKHS} & \multirow[t]{4}{*}{} & Observed & 0.914 & 0.020 & 4.720 & 9.806 \\
 &  &  & Linear Tilt & 0.913 & 0.023 & 4.588 & 8.097 \\
 &  &  & Local Tilt & 0.922 & 0.026 & 4.685 & 7.389 \\
 &  &  & Step Tilt & 0.913 & 0.022 & 4.584 & 8.059 \\
\cline{1-8} \cline{2-8} \cline{3-8}
\multirow[t]{12}{*}{Z-indexed} & \multirow[t]{4}{*}{Bins} & \multirow[t]{4}{*}{} & Observed & 0.904 & 0.021 & 4.227 & 8.520 \\
 &  &  & Linear Tilt & 0.904 & 0.025 & 4.236 & 8.186 \\
 &  &  & Local Tilt & 0.902 & 0.029 & 4.201 & 7.677 \\
 &  &  & Step Tilt & 0.904 & 0.024 & 4.202 & 7.971 \\
\cline{2-8} \cline{3-8}
 & \multirow[t]{4}{*}{Linear} & \multirow[t]{4}{*}{} & Observed & 0.906 & 0.020 & 4.046 & 7.553 \\
 &  &  & Linear Tilt & 0.907 & 0.026 & 4.114 & 7.563 \\
 &  &  & Local Tilt & 0.905 & 0.029 & 4.147 & 7.571 \\
 &  &  & Step Tilt & 0.906 & 0.024 & 4.094 & 7.557 \\
\cline{2-8} \cline{3-8}
 & \multirow[t]{4}{*}{RKHS} & \multirow[t]{4}{*}{} & Observed & 0.913 & 0.019 & 4.446 & 9.054 \\
 &  &  & Linear Tilt & 0.913 & 0.024 & 4.422 & 8.218 \\
 &  &  & Local Tilt & 0.912 & 0.029 & 4.455 & 7.764 \\
 &  &  & Step Tilt & 0.913 & 0.022 & 4.402 & 8.189 \\
\cline{1-8} \cline{2-8} \cline{3-8}
\multirow[t]{16}{*}{X-indexed} & \multirow[t]{16}{*}{} & \multirow[t]{4}{*}{Linear} & Observed & 0.906 & 0.028 & 4.904 & 15.400 \\
 &  &  & Linear Tilt & 0.917 & 0.032 & 5.415 & 17.546 \\
 &  &  & Local Tilt & 0.919 & 0.036 & 5.808 & 18.995 \\
 &  &  & Step Tilt & 0.913 & 0.031 & 5.313 & 17.559 \\
\cline{3-8}
 &  & \multirow[t]{4}{*}{Bins} & Observed & 0.907 & 0.027 & 20.055 & 164.856 \\
 &  &  & Linear Tilt & 0.918 & 0.031 & 25.255 & 213.098 \\
 &  &  & Local Tilt & 0.921 & 0.036 & 27.552 & 233.685 \\
 &  &  & Step Tilt & 0.915 & 0.030 & 25.166 & 213.356 \\
\cline{3-8}
 &  & \multirow[t]{4}{*}{RKHS} & Observed & 0.907 & 0.029 & 10.505 & 71.069 \\
 &  &  & Linear Tilt & 0.917 & 0.031 & 10.207 & 64.795 \\
 &  &  & Local Tilt & 0.919 & 0.036 & 7.377 & 33.785 \\
 &  &  & Step Tilt & 0.913 & 0.031 & 10.043 & 64.432 \\
\cline{3-8}
 &  & \multirow[t]{4}{*}{MLP} & Observed & 0.906 & 0.029 & inf & - \\
 &  &  & Linear Tilt & 0.918 & 0.030 & inf & - \\
 &  &  & Local Tilt & 0.919 & 0.036 & inf & - \\
 &  &  & Step Tilt & 0.914 & 0.029 & inf & - \\
\cline{1-8} \cline{2-8} \cline{3-8}
\bottomrule
\end{tabular}

}
\end{table}

\begin{table}[!th]
\caption{Results of the experiment using Dataset~2.}
    \label{tab:dataset2}
\centering
\resizebox{0.7\linewidth}{!}{
\begin{tabular}{llllrrrr}
\toprule
\multirow[c]{2}{*}{Radius} & \multirow[c]{2}{*}{IV Shift} & \multirow[c]{2}{*}{Radius Model} & \multirow[c]{2}{*}{Test Dist. of IVs} & \multicolumn{2}{c}{Cov. Ratio} & \multicolumn{2}{c}{Interval Length} \\
 &  & & &  Mean & SD & Mean & SD \\
\midrule
\multirow[t]{12}{*}{XZ-indexed} & \multirow[t]{4}{*}{Bins} & \multirow[t]{4}{*}{} & Observed & 0.901 & 0.020 & 2.751 & 0.210 \\
 &  &  & Linear Tilt & 0.898 & 0.024 & 2.800 & 0.246 \\
 &  &  & Local Tilt & 0.889 & 0.028 & 2.822 & 0.274 \\
 &  &  & Step Tilt & 0.900 & 0.021 & 2.786 & 0.235 \\
\cline{2-8} \cline{3-8}
 & \multirow[t]{4}{*}{Linear} & \multirow[t]{4}{*}{} & Observed & 0.914 & 0.023 & 2.858 & 0.239 \\
 &  &  & Linear Tilt & 0.915 & 0.025 & 2.951 & 0.282 \\
 &  &  & Local Tilt & 0.916 & 0.027 & 3.070 & 0.302 \\
 &  &  & Step Tilt & 0.917 & 0.023 & 2.944 & 0.269 \\
\cline{2-8} \cline{3-8}
 & \multirow[t]{4}{*}{RKHS} & \multirow[t]{4}{*}{} & Observed & 0.912 & 0.020 & inf & -  \\
 &  &  & Linear Tilt & 0.910 & 0.022 & inf & -  \\
 &  &  & Local Tilt & 0.906 & 0.025 & inf & -  \\
 &  &  & Step Tilt & 0.910 & 0.021 & inf & -  \\
\cline{1-8} \cline{2-8} \cline{3-8}
\multirow[t]{12}{*}{Z-indexed} & \multirow[t]{4}{*}{Bins} & \multirow[t]{4}{*}{} & Observed & 0.903 & 0.019 & 2.761 & 0.212 \\
 &  &  & Linear Tilt & 0.905 & 0.023 & 2.856 & 0.251 \\
 &  &  & Local Tilt & 0.904 & 0.026 & 2.956 & 0.281 \\
 &  &  & Step Tilt & 0.905 & 0.021 & 2.834 & 0.243 \\
\cline{2-8} \cline{3-8}
 & \multirow[t]{4}{*}{Linear} & \multirow[t]{4}{*}{} & Observed & 0.904 & 0.021 & 2.746 & 0.218 \\
 &  &  & Linear Tilt & 0.906 & 0.025 & 2.842 & 0.257 \\
 &  &  & Local Tilt & 0.898 & 0.028 & 2.880 & 0.279 \\
 &  &  & Step Tilt & 0.905 & 0.022 & 2.813 & 0.241 \\
\cline{2-8} \cline{3-8}
 & \multirow[t]{4}{*}{RKHS} & \multirow[t]{4}{*}{} & Observed & 0.914 & 0.020 & 2.906 & 0.230 \\
 &  &  & Linear Tilt & 0.916 & 0.023 & 3.014 & 0.300 \\
 &  &  & Local Tilt & 0.917 & 0.026 & 3.139 & 0.331 \\
 &  &  & Step Tilt & 0.917 & 0.022 & 2.988 & 0.276 \\
\cline{1-8} \cline{2-8} \cline{3-8}
\multirow[t]{16}{*}{X-indexed} & \multirow[t]{16}{*}{} & \multirow[t]{4}{*}{Linear} & Observed & 0.901 & 0.028 & 2.827 & 0.357 \\
 &  &  & Linear Tilt & 0.920 & 0.029 & 3.156 & 0.488 \\
 &  &  & Local Tilt & 0.920 & 0.036 & 3.298 & 0.541 \\
 &  &  & Step Tilt & 0.913 & 0.028 & 3.027 & 0.435 \\
\cline{3-8}
 &  & \multirow[t]{4}{*}{Bins} & Observed & 0.902 & 0.027 & 2.920 & 0.348 \\
 &  &  & Linear Tilt & 0.916 & 0.033 & 3.229 & 0.477 \\
 &  &  & Local Tilt & 0.918 & 0.037 & 3.427 & 0.550 \\
 &  &  & Step Tilt & 0.913 & 0.031 & 3.143 & 0.462 \\
\cline{3-8}
 &  & \multirow[t]{4}{*}{RKHS} & Observed & 0.904 & 0.027 & 2.855 & 0.330 \\
 &  &  & Linear Tilt & 0.917 & 0.031 & 3.133 & 0.444 \\
 &  &  & Local Tilt & 0.921 & 0.035 & 3.339 & 0.555 \\
 &  &  & Step Tilt & 0.914 & 0.029 & 3.056 & 0.425 \\
\cline{3-8}
 &  & \multirow[t]{4}{*}{MLP} & Observed & 0.903 & 0.032 & 3.149 & 0.523 \\
 &  &  & Linear Tilt & 0.917 & 0.036 & 3.527 & 0.739 \\
 &  &  & Local Tilt & 0.920 & 0.040 & 3.775 & 0.899 \\
 &  &  & Step Tilt & 0.914 & 0.034 & 3.388 & 0.600 \\
\cline{1-8} \cline{2-8} \cline{3-8}
\bottomrule
\end{tabular}

}
\end{table}

\begin{table}[!th]
\caption{Results of the experiment using Dataset~3.}
 \label{tab:dataset3}
\centering
\resizebox{0.7\linewidth}{!}{
\begin{tabular}{llllrrrr}
\toprule
\multirow[c]{2}{*}{Radius} & \multirow[c]{2}{*}{IV Shift} & \multirow[c]{2}{*}{Radius Model} & \multirow[c]{2}{*}{Test Dist. of IVs} & \multicolumn{2}{c}{Cov. Ratio} & \multicolumn{2}{c}{Interval Length} \\
 &  & & &  Mean & SD & Mean & SD \\
\midrule
\multirow[t]{12}{*}{XZ-indexed} & \multirow[t]{4}{*}{Bins} & \multirow[t]{4}{*}{} & Observed & 0.899 & 0.022 & 34.507 & 141.721 \\
 &  &  & Linear Tilt & 0.898 & 0.023 & 34.532 & 141.043 \\
 &  &  & Local Tilt & 0.899 & 0.025 & 34.710 & 142.243 \\
 &  &  & Step Tilt & 0.898 & 0.024 & 34.487 & 140.992 \\
\cline{2-8} \cline{3-8}
 & \multirow[t]{4}{*}{Linear} & \multirow[t]{4}{*}{} & Observed & 0.916 & 0.021 & 37.134 & 149.733 \\
 &  &  & Linear Tilt & 0.917 & 0.024 & 38.053 & 156.486 \\
 &  &  & Local Tilt & 0.917 & 0.028 & 38.470 & 159.830 \\
 &  &  & Step Tilt & 0.917 & 0.024 & 37.813 & 154.825 \\
\cline{2-8} \cline{3-8}
 & \multirow[t]{4}{*}{RKHS} & \multirow[t]{4}{*}{} & Observed & 0.912 & 0.022 & inf & -  \\
 &  &  & Linear Tilt & 0.913 & 0.026 & inf & -  \\
 &  &  & Local Tilt & 0.916 & 0.030 & inf & -  \\
 &  &  & Step Tilt & 0.913 & 0.026 & inf & -  \\
\cline{1-8} \cline{2-8} \cline{3-8}
\multirow[t]{12}{*}{Z-indexed} & \multirow[t]{4}{*}{Bins} & \multirow[t]{4}{*}{} & Observed & 0.902 & 0.023 & 34.581 & 143.175 \\
 &  &  & Linear Tilt & 0.903 & 0.024 & 35.391 & 149.331 \\
 &  &  & Local Tilt & 0.904 & 0.026 & 35.540 & 149.977 \\
 &  &  & Step Tilt & 0.903 & 0.024 & 35.012 & 146.496 \\
\cline{2-8} \cline{3-8}
 & \multirow[t]{4}{*}{Linear} & \multirow[t]{4}{*}{} & Observed & 0.911 & 0.021 & 37.132 & 156.724 \\
 &  &  & Linear Tilt & 0.910 & 0.024 & 38.008 & 163.341 \\
 &  &  & Local Tilt & 0.910 & 0.028 & 38.744 & 169.183 \\
 &  &  & Step Tilt & 0.911 & 0.023 & 37.830 & 162.051 \\
\cline{2-8} \cline{3-8}
 & \multirow[t]{4}{*}{RKHS} & \multirow[t]{4}{*}{} & Observed & 0.911 & 0.021 & 35.587 & 139.016 \\
 &  &  & Linear Tilt & 0.910 & 0.023 & 36.831 & 147.884 \\
 &  &  & Local Tilt & 0.910 & 0.026 & 37.425 & 151.207 \\
 &  &  & Step Tilt & 0.911 & 0.023 & 36.658 & 146.601 \\
\cline{1-8} \cline{2-8} \cline{3-8}
\multirow[t]{16}{*}{X-indexed} & \multirow[t]{16}{*}{} & \multirow[t]{4}{*}{Linear} & Observed & 0.895 & 0.035 & 1323.543 & 11129.639 \\
 &  &  & Linear Tilt & 0.909 & 0.033 & 1402.822 & 11599.136 \\
 &  &  & Local Tilt & 0.915 & 0.035 & 836.614 & 6013.192 \\
 &  &  & Step Tilt & 0.906 & 0.035 & 1302.644 & 10630.369 \\
\cline{3-8}
 &  & \multirow[t]{4}{*}{Bins} & Observed & 0.898 & 0.032 & 243.714 & 817.299 \\
 &  &  & Linear Tilt & 0.910 & 0.032 & 274.721 & 922.389 \\
 &  &  & Local Tilt & 0.918 & 0.035 & 293.849 & 1032.362 \\
 &  &  & Step Tilt & 0.910 & 0.033 & 271.293 & 907.443 \\
\cline{3-8}
 &  & \multirow[t]{4}{*}{RKHS} & Observed & 0.898 & 0.035 & 145.999 & 972.635 \\
 &  &  & Linear Tilt & 0.912 & 0.036 & 167.911 & 1126.497 \\
 &  &  & Local Tilt & 0.916 & 0.040 & 163.141 & 1088.062 \\
 &  &  & Step Tilt & 0.909 & 0.038 & 310.516 & 2531.610 \\
\cline{3-8}
 &  & \multirow[t]{4}{*}{MLP} & Observed & 0.902 & 0.033 & inf & - \\
 &  &  & Linear Tilt & 0.911 & 0.037 & inf & - \\
 &  &  & Local Tilt & 0.915 & 0.042 & inf & - \\
 &  &  & Step Tilt & 0.908 & 0.036 & inf & - \\
\cline{1-8} \cline{2-8} \cline{3-8}
\bottomrule
\end{tabular}
}
\end{table}

\subsection{Data-Generating Process}
We consider three synthetic NPIV designs of increasing difficulty. In all designs, the instruments are sampled from a uniform distribution on $\sqb{-1,1}^{d_Z}$, the endogenous regressors are nonlinear functions of the IVs and latent variables, and the outcome is generated as
\[
Y=h_0(X)+\text{confounding}+\sigma(X,Z)V,
\]
where the latent variables shared by $X$ and $Y$ induce endogeneity and $\sigma(X,Z)$ yields heteroskedastic predictive uncertainty.

\paragraph{Dataset~1 (\texorpdfstring{$d_X=d_Z=1$}{dX=dZ=1}).}
Let $Z\sim \mathrm{Unif}[-1,1]$ and let $U,E,V$ be independent standard normal variables. We generate
\[
X=1.05\sin(\pi Z)+0.55U+0.12E,
\qquad
h_0(X)=\frac{\sin(\pi X)}{1+0.45X^2},
\]
and
\[
Y=h_0(X)+0.55U+\sigma(X,Z)V,
\]
with
\[
\sigma(X,Z)=0.35+0.10|X|+0.12\frac{Z+1}{2}+0.08\p{1+\exp(-XZ)}^{-1}.
\]
This one-dimensional design is the basis for the visualization in Figure~\ref{fig:surface}.

\paragraph{Dataset~2 (\texorpdfstring{$d_X=3,d_Z=1$}{dX=3,dZ=1}).}
Let $Z\sim \mathrm{Unif}[-1,1]$ and let $U_1,U_2,U_3,E_1,E_2,E_3,V$ be mutually independent standard normal variables. We set
\[
\begin{aligned}
X_1&=\sin(\pi Z)+0.50U_1+0.10E_1,\\
X_2&=Z^2-\frac13+0.42U_2+0.10E_2,\\
X_3&=\cos(\pi Z)+0.24(U_1+U_3)+0.10E_3,
\end{aligned}
\]
\[
h_0(X)=\sin(X_1)+0.45X_2^2-0.35X_1X_3+0.40\cos(X_3),
\]
and
\[
Y=h_0(X)+(0.40U_1-0.30U_2+0.22U_3)+\sigma(X,Z)V,
\]
where
\[
\sigma(X,Z)=0.32+0.08|X_1|+0.07X_2^2+0.10\frac{Z+1}{2}+0.05\max\cb{X_3Z,0}.
\]

\paragraph{Dataset~3 (\texorpdfstring{$d_X=d_Z=3$}{dX=dZ=3}).}
Let $Z=(Z_1,Z_2,Z_3)$ have independent coordinates distributed as $\mathrm{Unif}[-1,1]$, and let $U_1,U_2,U_3,E_1,E_2,E_3,V$ be mutually independent standard normal variables. We generate
\[
\begin{aligned}
X_1&=\sin\p{\pi(Z_1+0.30Z_2)}+0.42U_1+0.10E_1,\\
X_2&=Z_2Z_3+0.25Z_1^2+0.40U_2+0.10E_2,\\
X_3&=\cos(\pi Z_3)+0.22Z_1-0.18Z_2+0.28(U_1+U_3)+0.10E_3,
\end{aligned}
\]
\[
h_0(X)=0.60\sin(X_1)+0.38X_2^2-0.25X_1X_3+0.48\cos(X_3)+0.18X_1X_2,
\]
and
\[
Y=h_0(X)+(0.36U_1-0.24U_2+0.18U_3)+\sigma(X,Z)V,
\]
with
\[
\sigma(X,Z)=0.30+0.03(X_1^2+X_2^2)+0.07\max\{X_3,0\}+0.07\frac{Z_1+1}{2}+0.05\frac{Z_2Z_3+1}{2}.
\]
This is the most difficult design, because both the structural function and the heteroskedastic scale depend on higher-dimensional nonlinear interactions.

\subsection{Experimental Setup}
We set the nominal coverage level to $1-\alpha=0.9$ and run $100$ Monte Carlo replications for each synthetic design. Each replication uses $n_{\mathrm{train}}=1000$ observations to fit the NPIV base learner, $n_{\mathrm{cal}}=200$ observations for conformal calibration, and $n_{\mathrm{test}}=1000$ observations for evaluation. Across all methods, the center $\widehat h$ is a common series-2SLS NPIV estimator with cubic polynomial bases in both $X$ and $Z$.

For the exact classes $\calT_{XZ}$ and $\calT_Z$, we consider three finite-dimensional shift families. The \textit{Bins} family uses a four-bin basis defined on a one-dimensional projection. For scalar IVs, this projection is $Z$ itself. For vector IVs, it is the standardized first principal component of the pooled training and calibration IVs. The \textit{Linear} family uses an affine basis with an intercept, and the \textit{RKHS} family uses a Gaussian-kernel dictionary with four landmarks and kernel parameter $\gamma=0.2$. For the $X$-indexed class $\calT_X$, we consider four radius models: a linear basis, a six-bin basis, an RKHS basis with four landmarks and $\gamma=0.2$, and an MLP radius model. The conditional density ratio in the $X$-indexed learner is estimated by a neural network with hidden layers $(32,32)$ using two-fold cross-fitted training scores. We then split the calibration sample evenly into a shape-learning half and an independent final recalibration half. The $X$-indexed learner uses $\lambda=50$, $\kappa=0.05$, and, for the neural radius model, hidden layers $(32,32)$.

To evaluate robustness under IV shifts, we report four test laws. Let $u(z)$ be the standardized one-dimensional projection of $z$ computed from the pooled training and calibration IVs, and when $Z$ is multivariate, $u(z)$ is the first principal component after standardization. The raw weights are
\[
w_{\mathrm{obs}}(u)=1,\qquad
w_{\mathrm{lin}}(u)=\max\cb{1+0.95u,0.05},
\]
\[
w_{\mathrm{step}}(u)=
\begin{cases}
e, & u>0,\\
1, & u\le 0,
\end{cases}
\qquad
w_{\mathrm{loc}}(u)=0.20+1.60\exp\!\sqb{-\frac12\p{\frac{u-0.75}{0.35}}^2}.
\]
For each method and scenario, the reported coverage ratio and interval length are weighted test-sample averages computed with the normalized versions of these weights. An entry of \texttt{inf} in the length column means that at least one replication produced an unbounded interval, which makes the across-replication mean infinite.

\subsection{Results}
Tables~\ref{tab:dataset1}--\ref{tab:dataset3} summarize the synthetic experiments. Overall, the exact classes $\calT_{XZ}$ and $\calT_Z$ behave as intended for the Bins and Linear shift families: their coverage ratios stay close to the nominal level across the observed law and the three shifted IV laws. In Dataset~1, exact interval lengths are around $4$ for both classes, with the Linear specification slightly shorter than Bins and RKHS. In Dataset~2, the same pattern remains, and the exact finite-length intervals are even shorter, around $2.7$ to $3.1$. In Dataset~3, the exact procedures remain finite for Bins and Linear, but the intervals become much longer, reflecting the more difficult three-dimensional nonlinear design.

The RKHS proxy is more fragile in the exact classes. In Dataset~1, both $\calT_{XZ}$ and $\calT_Z$ with RKHS remain finite and are mildly conservative. In Dataset~2, the $Z$-indexed RKHS procedure is still well behaved, but the $(X,Z)$-indexed RKHS procedure yields unbounded intervals in at least one replication. In Dataset~3, the same contrast becomes sharper: the $Z$-indexed RKHS method remains finite, whereas the $(X,Z)$-indexed RKHS method again produces \texttt{inf}. This pattern is consistent with the fact that richer exact shift classes are harder to calibrate stably as the dimension of the conditioning variable grows.

The $X$-indexed results illustrate the trade-off between interpretability and difficulty. In Dataset~1, the Linear $X$-indexed model is the most stable $C(X)$ rule, with coverage near the nominal level and only a moderate increase in length relative to the exact methods. The RKHS $X$-indexed model is workable but more variable, while the Bins and especially the MLP radius models are unstable. In Dataset~2, the Linear, Bins, and RKHS $X$-indexed models all remain finite and near nominal, showing that fixed-shift recalibration can work well in a moderate-dimensional design. In Dataset~3, however, all $X$-indexed procedures become much longer and more variable, especially the Linear and MLP models. This is exactly the difficult regime for the target class $\calT_X$: the final interval is required to depend only on $X$, even though the predictive uncertainty in the data-generating process varies with both $X$ and $Z$.

\subsection{Visualization of Prediction Intervals}
Figure~\ref{fig:surface} visualizes the fitted intervals for Dataset~1 under the observed IV distribution using a single replication with seed $2026$. For the exact procedures, we use the Linear shift family, and for the $X$-indexed procedure we use the Linear radius model. The top row plots the lower and upper interval surfaces over a $41\times 41$ grid in $(x,z)$. The middle row fixes $x=0$ and plots the $Y$--$Z$ slices, while the bottom row fixes $z=0$ and plots the $Y$--$X$ slices.

Note that large or infinite-length prediction intervals do not necessarily imply poor performance. If one of the upper or lower bounds is large or infinite, the interval length becomes infinite. However, the opposite bound can still be finite, and if so, we can still use such prediction intervals in practice. A more important metric is the coverage ratio.

\begin{figure}
    \centering
    \includegraphics[width=0.9\linewidth]{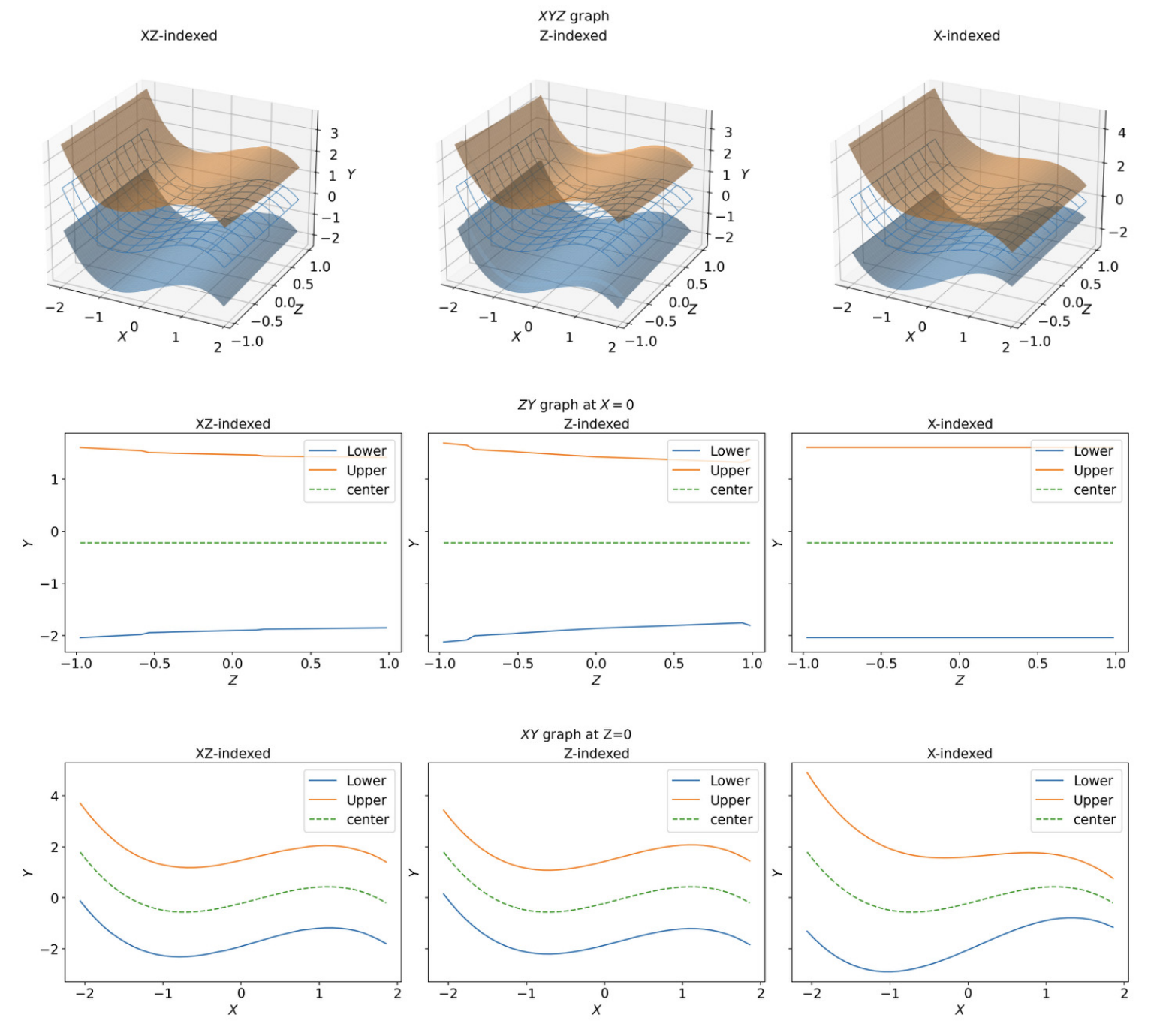}
    \caption{Visualization of prediction intervals for Dataset~1 under the observed IV distribution. The $(X,Z)$-indexed and $Z$-indexed procedures use the Linear shift family, and the $X$-indexed procedure uses the Linear radius model. The top row shows the interval surfaces over $(x,z)$, the middle row shows $Y$--$Z$ slices at $x=0$, and the bottom row shows $Y$--$X$ slices at $z=0$.}
    \label{fig:surface}
\end{figure}

The figure makes the distinction among the three radius classes visually transparent. In the middle row, the $(X,Z)$-indexed and $Z$-indexed intervals vary with $z$, whereas the $X$-indexed interval is flat in $z$ by construction. In the bottom row, all three methods vary with $x$, but the width profiles differ: the exact methods adapt to IV-indexed uncertainty, while the $X$-indexed interval absorbs that uncertainty into a single $X$-only radius. In Dataset~1, the $(X,Z)$-indexed and $Z$-indexed surfaces are quite similar, which suggests that much of the relevant heteroskedasticity can already be captured by allowing the radius to depend on the IV alone.

\subsection{Discussion}
The simulation evidence supports the main conceptual message of the paper. The exact $Z$-indexed class $\calT_Z$ is the most reliable practical exact procedure: it consistently delivers near-nominal coverage under the shifted IV laws while keeping the interval lengths comparable to, and often slightly shorter than, those of the broader exact class $\calT_{XZ}$. The $X$-indexed target $\calT_X$ remains substantively appealing, because it yields intervals of the form $C(X)$, but its finite-sample performance depends much more strongly on the radius model and on the complexity of the design. In low- and moderate-dimensional settings, simple $X$-indexed models such as the Linear basis can work well. In harder settings, the price of insisting on $Z$-free final intervals can be substantial.

These results also clarify how to interpret richer model classes in practice. RKHS-based exact classes can become unstable in more difficult designs, and overly flexible $X$-indexed radius models can produce very long intervals. Accordingly, the finite-dimensional Bins and Linear exact classes are the most robust baselines in the current experiments, while richer classes are better viewed as sensitivity analyses rather than default choices.

\section{Empirical Analysis}
We next study two empirical IV datasets using the same nominal level, shift scenarios, and evaluation metrics as in Section~7. In the \textsc{CigarettesSW} data \citep{Stock2007introductionto}, we take
\[
Y=\log\p{\texttt{packs}},\qquad
X=\log\p{\texttt{price}/\texttt{cpi}},\qquad
Z=\log\p{\texttt{taxs}/\texttt{cpi}},
\]
and include $\log\p{\texttt{income}/\texttt{population}/\texttt{cpi}}$ and a 1995 year indicator as exogenous controls. Each repetition uses a year-stratified split with $24$ training observations, $40$ calibration observations, and $32$ test observations. In the \textsc{CollegeDistance} data \citep{Stock2007introductionto,Rouse1995democratizationor}, we take
\[
Y=\log\p{\max\cb{\texttt{wage},10^{-2}}},\qquad
X=\texttt{education},\qquad
Z=\texttt{distance},
\]
and use \texttt{score}, \texttt{tuition}, \texttt{unemp}, and demographic dummy variables as controls. To keep the exact conditional procedures computationally manageable, each repetition first subsamples $1800$ observations and then uses a $600/600/600$ training, calibration, and test split. All empirical results are averaged over $100$ random replications.

\begin{table}[!th]
\caption{Results of the experiment using the CigarettesSW dataset.}
\label{tab:CigarettesSW}
    \centering
\resizebox{0.7\linewidth}{!}{
\begin{tabular}{llllrrrr}
\toprule
\multirow[c]{2}{*}{Radius} & \multirow[c]{2}{*}{IV Shift} & \multirow[c]{2}{*}{Radius Model} & \multirow[c]{2}{*}{Test Dist. of IVs} & \multicolumn{2}{c}{Cov. Ratio} & \multicolumn{2}{c}{Interval Length} \\
 &  & & &  Mean & SD & Mean & SD \\
\midrule
\multirow[t]{12}{*}{XZ-indexed} & \multirow[t]{4}{*}{Bins} & \multirow[t]{4}{*}{} & Observed & 0.911 & 0.055 & 1.489 & 2.392 \\
 &  &  & Linear Tilt & 0.918 & 0.060 & 1.282 & 1.511 \\
 &  &  & Local Tilt & 0.918 & 0.070 & 1.335 & 1.621 \\
 &  &  & Step Tilt & 0.919 & 0.057 & 1.360 & 1.816 \\
\cline{2-8} \cline{3-8}
 & \multirow[t]{4}{*}{Linear} & \multirow[t]{4}{*}{} & Observed & 0.928 & 0.051 & inf & -  \\
 &  &  & Linear Tilt & 0.935 & 0.056 & inf & -  \\
 &  &  & Local Tilt & 0.936 & 0.064 & inf & -  \\
 &  &  & Step Tilt & 0.934 & 0.051 & inf & -  \\
\cline{2-8} \cline{3-8}
 & \multirow[t]{4}{*}{RKHS} & \multirow[t]{4}{*}{} & Observed & 0.957 & 0.040 & inf & -  \\
 &  &  & Linear Tilt & 0.960 & 0.041 & inf & -  \\
 &  &  & Local Tilt & 0.969 & 0.038 & inf & -  \\
 &  &  & Step Tilt & 0.962 & 0.040 & inf & -  \\
\cline{1-8} \cline{2-8} \cline{3-8}
\multirow[t]{12}{*}{Z-indexed} & \multirow[t]{4}{*}{Bins} & \multirow[t]{4}{*}{} & Observed & 0.915 & 0.051 & 1.524 & 2.545 \\
 &  &  & Linear Tilt & 0.919 & 0.058 & 1.299 & 1.668 \\
 &  &  & Local Tilt & 0.916 & 0.069 & 1.337 & 1.756 \\
 &  &  & Step Tilt & 0.916 & 0.059 & 1.374 & 1.960 \\
\cline{2-8} \cline{3-8}
 & \multirow[t]{4}{*}{Linear} & \multirow[t]{4}{*}{} & Observed & 0.920 & 0.054 & 1.476 & 2.385 \\
 &  &  & Linear Tilt & 0.928 & 0.060 & 1.340 & 1.774 \\
 &  &  & Local Tilt & 0.926 & 0.072 & 1.329 & 1.633 \\
 &  &  & Step Tilt & 0.926 & 0.055 & 1.382 & 1.956 \\
\cline{2-8} \cline{3-8}
 & \multirow[t]{4}{*}{RKHS} & \multirow[t]{4}{*}{} & Observed & 0.954 & 0.048 & inf & -  \\
 &  &  & Linear Tilt & 0.954 & 0.054 & inf & -  \\
 &  &  & Local Tilt & 0.960 & 0.062 & inf & -  \\
 &  &  & Step Tilt & 0.956 & 0.053 & inf & -  \\
\cline{1-8} \cline{2-8} \cline{3-8}
\multirow[t]{16}{*}{X-indexed} & \multirow[t]{16}{*}{} & \multirow[t]{4}{*}{Linear} & Observed & 0.908 & 0.080 & 1.587 & 2.982 \\
 &  &  & Linear Tilt & 0.965 & 0.045 & inf & -  \\
 &  &  & Local Tilt & 0.997 & 0.017 & inf & -  \\
 &  &  & Step Tilt & 0.942 & 0.064 & 2.889 & 12.223 \\
\cline{3-8}
 &  & \multirow[t]{4}{*}{Bins} & Observed & 0.908 & 0.078 & 30.228 & 218.774 \\
 &  &  & Linear Tilt & 0.963 & 0.050 & inf & -  \\
 &  &  & Local Tilt & 0.998 & 0.014 & inf & -  \\
 &  &  & Step Tilt & 0.940 & 0.065 & 75.474 & 520.323 \\
\cline{3-8}
 &  & \multirow[t]{4}{*}{RKHS} & Observed & 0.915 & 0.080 & 1.656 & 3.127 \\
 &  &  & Linear Tilt & 0.971 & 0.042 & inf & -  \\
 &  &  & Local Tilt & 0.999 & 0.007 & inf & -  \\
 &  &  & Step Tilt & 0.946 & 0.067 & 2.327 & 5.911 \\
\cline{3-8}
 &  & \multirow[t]{4}{*}{MLP} & Observed & 0.909 & 0.076 & 104.428 & 1022.857 \\
 &  &  & Linear Tilt & 0.961 & 0.050 & inf & -  \\
 &  &  & Local Tilt & 0.999 & 0.009 & inf & -  \\
 &  &  & Step Tilt & 0.942 & 0.064 & 110.968 & 1075.274 \\
\cline{1-8} \cline{2-8} \cline{3-8}
\bottomrule
\end{tabular}

}
\end{table}

\begin{table}[!th]
\caption{Results of the experiment using the CollegeDistance dataset.}
\label{tab:CollegeDistance}
\centering
\resizebox{0.7\linewidth}{!}{
\begin{tabular}{llllrrrr}
\toprule
\multirow[c]{2}{*}{Radius} & \multirow[c]{2}{*}{IV Shift} & \multirow[c]{2}{*}{Radius Model} & \multirow[c]{2}{*}{Test Dist. of IVs} & \multicolumn{2}{c}{Cov. Ratio} & \multicolumn{2}{c}{Interval Length} \\
 &  & & &  Mean & SD & Mean & SD \\
\midrule
\multirow[t]{12}{*}{XZ-indexed} & \multirow[t]{4}{*}{Bins} & \multirow[t]{4}{*}{} & Observed & 0.901 & 0.017 & 2.638 & 4.845 \\
 &  &  & Linear Tilt & 0.893 & 0.020 & 2.521 & 4.446 \\
 &  &  & Local Tilt & 0.882 & 0.028 & 2.386 & 4.043 \\
 &  &  & Step Tilt & 0.891 & 0.021 & 2.502 & 4.402 \\
\cline{2-8} \cline{3-8}
 & \multirow[t]{4}{*}{Linear} & \multirow[t]{4}{*}{} & Observed & 0.903 & 0.017 & 2.668 & 4.744 \\
 &  &  & Linear Tilt & 0.903 & 0.019 & 2.626 & 4.599 \\
 &  &  & Local Tilt & 0.904 & 0.025 & 2.570 & 4.429 \\
 &  &  & Step Tilt & 0.903 & 0.021 & 2.622 & 4.605 \\
\cline{2-8} \cline{3-8}
 & \multirow[t]{4}{*}{RKHS} & \multirow[t]{4}{*}{} & Observed & 0.904 & 0.015 & inf & -  \\
 &  &  & Linear Tilt & 0.908 & 0.017 & inf & -  \\
 &  &  & Local Tilt & 0.915 & 0.024 & inf & -  \\
 &  &  & Step Tilt & 0.907 & 0.019 & inf & -  \\
\cline{1-8} \cline{2-8} \cline{3-8}
\multirow[t]{12}{*}{Z-indexed} & \multirow[t]{4}{*}{Bins} & \multirow[t]{4}{*}{} & Observed & 0.899 & 0.018 & 2.671 & 5.017 \\
 &  &  & Linear Tilt & 0.900 & 0.019 & 2.619 & 4.750 \\
 &  &  & Local Tilt & 0.900 & 0.023 & 2.558 & 4.492 \\
 &  &  & Step Tilt & 0.899 & 0.020 & 2.608 & 4.712 \\
\cline{2-8} \cline{3-8}
 & \multirow[t]{4}{*}{Linear} & \multirow[t]{4}{*}{} & Observed & 0.898 & 0.018 & 2.573 & 4.519 \\
 &  &  & Linear Tilt & 0.899 & 0.020 & 2.555 & 4.447 \\
 &  &  & Local Tilt & 0.900 & 0.025 & 2.534 & 4.378 \\
 &  &  & Step Tilt & 0.898 & 0.020 & 2.553 & 4.448 \\
\cline{2-8} \cline{3-8}
 & \multirow[t]{4}{*}{RKHS} & \multirow[t]{4}{*}{} & Observed & 0.901 & 0.017 & inf & -  \\
 &  &  & Linear Tilt & 0.904 & 0.018 & inf & -  \\
 &  &  & Local Tilt & 0.908 & 0.024 & inf & -  \\
 &  &  & Step Tilt & 0.903 & 0.019 & inf & -  \\
\cline{1-8} \cline{2-8} \cline{3-8}
\multirow[t]{16}{*}{X-indexed} & \multirow[t]{16}{*}{} & \multirow[t]{4}{*}{Linear} & Observed & 0.896 & 0.019 & 151.561 & 1492.597 \\
 &  &  & Linear Tilt & 0.901 & 0.020 & 154.493 & 1521.789 \\
 &  &  & Local Tilt & 0.906 & 0.026 & 151.402 & 1490.779 \\
 &  &  & Step Tilt & 0.901 & 0.021 & 140.837 & 1385.177 \\
\cline{3-8}
 &  & \multirow[t]{4}{*}{Bins} & Observed & 0.895 & 0.021 & 93.181 & 506.650 \\
 &  &  & Linear Tilt & 0.900 & 0.022 & 97.514 & 537.866 \\
 &  &  & Local Tilt & 0.905 & 0.028 & 100.555 & 545.249 \\
 &  &  & Step Tilt & 0.900 & 0.023 & 95.276 & 521.839 \\
\cline{3-8}
 &  & \multirow[t]{4}{*}{RKHS} & Observed & 0.895 & 0.019 & 79.233 & 770.463 \\
 &  &  & Linear Tilt & 0.901 & 0.021 & 80.615 & 784.204 \\
 &  &  & Local Tilt & 0.906 & 0.027 & 78.894 & 767.130 \\
 &  &  & Step Tilt & 0.901 & 0.021 & 73.547 & 713.580 \\
\cline{3-8}
 &  & \multirow[t]{4}{*}{MLP} & Observed & 0.895 & 0.020 & 2.544 & 4.955 \\
 &  &  & Linear Tilt & 0.901 & 0.021 & 2.549 & 4.962 \\
 &  &  & Local Tilt & 0.905 & 0.028 & 2.564 & 5.063 \\
 &  &  & Step Tilt & 0.901 & 0.022 & 2.543 & 4.938 \\
\cline{1-8} \cline{2-8} \cline{3-8}
\bottomrule
\end{tabular}
}
\end{table}

Tables~\ref{tab:CigarettesSW} and \ref{tab:CollegeDistance} show a sharper contrast between stable and unstable specifications than in the synthetic designs. In \textsc{CigarettesSW}, the exact Bins procedures and the exact $Z$-indexed Linear procedure remain finite and close to the nominal coverage level, whereas the richer exact classes frequently produce unbounded intervals. For the $X$-indexed class, the Observed and Step-Tilt scenarios are still workable for the Linear and RKHS radius models, but the Linear-Tilt and Local-Tilt scenarios often lead to \texttt{inf}. This is the small-sample regime in which aggressive reweightings of the IV distribution are hardest to support, so coarse exact shift classes are the most reliable practical choices.

In \textsc{CollegeDistance}, the exact Bins and Linear procedures are much more stable. Their coverage ratios stay close to the nominal level under all four scenarios, and their interval lengths are around $2.5$ to $2.7$. By contrast, the exact RKHS procedures again produce \texttt{inf}, which indicates that the richer exact shift class is still too unstable in this empirical design. The most striking pattern appears in the $X$-indexed class: the Linear, Bins, and RKHS radius models produce very long intervals, whereas the MLP radius model remains stable, with coverage around $0.895$ to $0.905$ and interval length around $2.54$ to $2.56$. Thus, in the larger empirical dataset, a flexible nonlinear $q(x)$ can be essential for obtaining a practical $C(X)$ interval.

Taken together, the empirical results complement the simulations. Exact IV-specific procedures based on Bins or Linear shift classes are the most robust default options across datasets. The $X$-indexed target can be competitive, but only when the radius model matches the complexity of the application and the final fixed-shift recalibration remains numerically stable.

\section{Conclusion}
We proposed IV-CCP for constructing prediction intervals in NPIV with a distribution-free finite-sample coverage guarantee. Our proposed prediction interval construction allows three types of radii, $(X, Z)$-indexed, $Z$-indexed, and $X$-indexed, whose classes are denoted by $\calT_{XZ}$, $\calT_Z$, and $\calT_X$, respectively. The broad benchmark class is $\calT_{XZ}$, which produces intervals whose radius may vary jointly with $(X,Z)$ and admits exact finite-sample coverage over joint shifts. The exact IV-specific class is $\calT_Z$, which keeps the center equal to $\widehat h(X)$ and yields IV-CCP, an exact finite-sample procedure over policy-relevant IV shifts. The class $\calT_X$ produces intervals $C(X)$, which may be the most natural operational target when the final prediction rule should depend on the structural regressor alone. For these approaches, we provide theoretical guarantees and confirm their validity through simulation studies and empirical analyses.

\bibliography{arXiv2.bbl}

\onecolumn
\appendix

\section{Detailed Related Work}\label{appdx:detailed_related_work}
The contribution of this paper is easiest to locate once the method is decomposed into the objects separated in the main text. The first object is structural estimation under endogeneity. This is the step that produces the base learner $\widehat h$ used in the score construction \eqref{eq:scores}. The second object is the radius class $\calT$, which determines whether the interval is allowed to vary with $(X,Z)$, only with $Z$, or only with $X$. The third object is the shift class $\calF$, which determines the family of the probability laws over which the coverage constraint \eqref{eq:moments_F} is required to hold. The related work below follows this decomposition, because it makes the connection to the rest of the paper transparent.

\subsection{Structural Models and the Moment Restrictions}
The econometric language behind our setup is the language of moment restrictions. In classical linear IV models, exogeneity is expressed through orthogonality conditions between the IV and the structural error, and overidentified models are naturally handled by GMM \citep{Sargan1958theestimation,Hansen1982largesample}. More broadly, \citet{Chamberlain1987asymptoticefficiency} studies estimation under conditional moment restrictions and provides the semiparametric framework that underlies many later developments in IV and NPIV. This background is directly relevant for our paper, because the structural equation \eqref{eq:npiv} and the conditional moment restriction \eqref{eq:cmc} are special cases of that general framework.

This literature also clarifies what our paper is not trying to do. In linear IV and GMM, the inferential target is typically a structural parameter or a finite-dimensional functional identified by the moments. In our setting, the structural learner still comes from that tradition, but the final target is a prediction interval for a future outcome. This distinction matters throughout the paper. The structural literature explains how to estimate $h_0$, whereas our prediction layer uses the resulting residual scores to enforce predictive validity.

\subsection{NPIV and Estimation under the Conditional Moment Restriction}
NPIV extends IV from finite-dimensional linear models to an unknown structural function. Foundational papers show how identification can be written as an inverse problem under a conditional moment restriction, and how estimation therefore requires regularization \citep{Newey2003instrumentalvariable,Ai2003efficientestimation,Darolles2011nonparametricinstrumental,Carrasco2007linearinverse}. The ill-posedness of NPIV, and its implications for convergence rates and smoothing choices, are central in \citet{Hall2005nonparametricmethods}. Related semiparametric developments, such as \citet{Blundell2007seminonparametriciv}, and broader conditional-moment formulations, such as \citet{Chen2012estimationof} and \citet{Otsu2011empiricallikelihood}, place NPIV inside a larger class of problems where the unknown object is a function and the identifying information is a conditional moment.

This literature is the natural source of the base learner in our algorithm. Neither the joint class $\calT_{XZ}$ nor the IV-indexed class $\calT_Z$ nor the $X$-indexed class $\calT_X$ prescribes a specific NPIV estimator. Any estimator that outputs a structural prediction rule $\widehat h(X)$ can be used in the training step, including sieve, minimum-distance, Tikhonov, and related regularized procedures.

\subsection{Asymptotic Inference in NPIV}
A substantial econometric literature studies inference for the structural function or its functionals in NPIV and related conditional moment models. Representative results include asymptotic normality for NPIV estimators \citep{Horowitz2007asymptoticnormality}, uniform confidence bands \citep{Horowitz2012uniformconfidence}, sieve Wald and QLR inference for semi/nonparametric conditional moment models \citep{Chen2015sievewald}, and sup-norm rates with uniform inference for nonlinear functionals \citep{Chen2018optimalsupnorm}. These papers are closely connected to our theory section, because they clarify how ill-posedness, approximation error, and regularization affect the quality of structural estimation.

At the same time, the inferential target is fundamentally different. Those papers develop confidence procedures for $h_0$ or for functionals of $h_0$, typically under asymptotic approximations and structural regularity conditions. By contrast, the object of this paper is the future outcome $Y_{n+1}$, and the guarantee in Theorem~\ref{thm:shift_robust} is finite-sample and distribution-free conditional on the training split. The connection to the NPIV inference literature appears through efficiency rather than through asymptotic Gaussianity. Proposition~\ref{prop:stability} and Theorem~\ref{thm:length} show that a smaller structural estimation error produces a shorter interval in the exact IV-indexed class, so advances in NPIV estimation feed directly into the efficiency of the conformal layer.

\subsection{Machine learning estimators for IV and conditional moments}
A recent machine learning literature develops flexible learners for IV and conditional moment problems. One group of methods extends the two-stage logic of 2SLS. \citet{Hartford2017deepiv} estimates the first-stage treatment distribution with neural networks and integrates the second-stage predictor over that estimated distribution. \citet{Singh2019kernelinstrumental} develops a kernel version of nonlinear IV regression in RKHS, and \citet{Xu2021learningdeep} learns deep representations for IVs and treatments before the two regression stages. These methods are directly compatible with our framework, because they produce exactly the object needed in \eqref{eq:scores}; that is, a base learner $\widehat h$.

Another group of methods formulates conditional moment estimation as a minimax or adversarial problem. This includes adversarial GMM \citep{Lewis2018adversarialgeneralized}, DeepGMM \citep{Bennett2019deepgeneralized}, the variational method of moments \citep{Bennett2020thevariational}, DualIV \citep{Muandet2020dualinstrumental}, neural min-max estimators for structural equations \citep{Liao2020provablyefficient}, and minimax-optimal procedures for IV regression \citep{Dikkala2020minimaxestimation}. The conceptual connection to our paper is important but subtle. In that literature, the adversarial or critic class is used to estimate $h_0$ by searching for violated moments. In our paper, the class $\calF$ is not part of the structural estimator. It encodes the test environments over which the predictive coverage requirement \eqref{eq:moments_F} must hold.

\subsection{Conformal Prediction}
Classical conformal prediction provides exact finite-sample marginal coverage under exchangeability \citep{Vovk2005algorithmiclearning}. A large literature improves efficiency by adapting interval length to heteroskedasticity or local structure, for example through conformalized quantile regression \citep{Romano2019conformalizedquantile}, distributional conformal prediction \citep{Chernozhukov2021distributionalconformal}, and localized conformal prediction \citep{Guan2022localizedconformal}. These papers matter for our setup because they show how conformal methods can wrap around arbitrary black-box predictors while preserving finite-sample marginal validity.

For our paper, however, the central issue is not only efficiency under marginal validity. It is how to move toward a validity notion indexed by the IV. Early work such as \citet{Vovk2012conditionalvalidity} studies different notions of conditional validity for conformal predictors, while \citet{Lei2013distributionfree} and \citet{Barber2020thelimits} show why exact distribution-free finite-sample conditional coverage of the form \eqref{eq:exact_conditional} cannot generally be achieved with nontrivial intervals. This impossibility result is the reason the main text replaces exact conditional coverage by the relaxed moment condition \eqref{eq:moments_F}.

Several recent strands of the literature study structured relaxations that are closer to what we need. When the test environment is a single specified test distribution, weighted conformal prediction provides valid coverage under covariate shift if the likelihood ratio is known or can be estimated \citep{Tibshirani2019conformalprediction}. When the shift is driven by a finite set of groups, \citet{Bhattacharyya2025groupweighted} studies group-weighted conformal prediction. A related line studies subgroup or multivalid guarantees over finite collections of sets \citep{Jung2023batchmultivalid,Bastani2022practicaladversarial}. These papers are directly relevant for our design of the shift class $\calF$. They also clarify the main distinction from our setting. We do not target one known test distribution. Instead, in the exact IV-indexed class $\calT_Z$ we require coverage to hold simultaneously over the family of IV tilts encoded by $\calF$.

The closest methodological antecedent is \citet{Gibbs2025conformalprediction}. That paper reformulates conditional validity as coverage over a class of covariate shifts and develops the augmented quantile-regression calibration algorithm that underlies our exact constructions for $\calT_{XZ}$ and $\calT_Z$. OOur paper specializes that framework to IV settings, treats the IV $Z$ as the variable indexing the protected shifts, and places this exact conformal layer alongside the importance-weighted $X$-indexed formulation of \citet{Kato2022learningcausal}.

\subsection{Position of this Paper}
Putting these strands of the literature together makes the contribution of the paper precise. Relative to the IV and NPIV literatures, we do not propose a new identification strategy or a new estimator of $h_0$. Relative to the conformal literature, we do not target marginal coverage only, and we do not assume a single known target distribution except in the $X$-indexed fixed-shift recalibration step. Instead, the main text organizes predictive inference in NPIV around three radius classes presented in the order $\calT_{XZ}$, $\calT_Z$, and $\calT_X$.

The class $\calT_{XZ}$ is the broad exact benchmark, because the variable indexing the radius and the variable indexing the shift family coincide. The class $\calT_Z$ is the exact IV-specific conformal class and the first main focus of the paper. The class $\calT_X$ is the target $C(X)$ and the second main focus. In that class, the paper combines importance-weighted conditional-moment learning with a fixed-shift weighted conformal recalibration layer. This is the sense in which the paper sits between structural estimation and predictive inference while being explicit about the difference between the prediction interval one ultimately wants and the exact subclasses current finite-sample conformal machinery can certify.

\section{Technical Discussion of Candidate Shift Classes}\label{appdx:shiftclasses}
This appendix records several concrete constructions that convert rich candidate shift families into the finite-dimensional feature map $\phi$ required by the exact conformal procedures for $\calT_{XZ}$ and $\calT_Z$. The common pattern is to start from an ideal class of nonnegative density ratios and then choose a finite dictionary whose span approximates that ideal class. This is directly parallel to the way sieve and low-rank approximations are used in the NPIV literature, including the appendices of \citet{Dikkala2020minimaxestimation}. The difference is that the present paper approximates a class of distribution shifts in the test environment, not a class of structural functions.

\subsection{Nested Approximations}
Let $\phi^{(r)}\colon \calZ\to\bbR^{d_r}$ denote a sequence of feature maps such that
\[
\calG_r=\cb{z\mapsto \beta^\top\phi^{(r)}(z): \beta\in\bbR^{d_r}}
\]
is nested in $r$, that is, $\calG_r\subset \calG_{r+1}$. Define the corresponding nonnegative shift classes by $\calF_r=\calG_r\cap\cb{f\ge 0}$. Then the ambiguity sets are nested as well,
\[
\calQ(\calF_r)\subset \calQ(\calF_{r+1}).
\]
Thus increasing $r$ strengthens the protected family of the test environments. For every fixed $r$, Theorems~\ref{thm:joint_shift_robust} and \ref{thm:shift_robust} remain exact and finite-sample. The price of increasing $r$ appears through the dimension term in Theorem~\ref{thm:length}.

\subsection{Discretization and Multiscale Partitions}
Suppose for simplicity that $Z$ is scalar and supported on an interval $\sqb{a,b}$. Let
\[
a=u_0<u_1<\cdots<u_G=b
\]
be a grid, and define cells $I_g=\sqb{u_{g-1},u_g}$ for $g=1,\dots,G$. The most basic discretization takes
\[
\phi(z)=\p{1,\mathbbm{1}\sqb{z\in I_1},\dots,\mathbbm{1}\sqb{z\in I_G}}.
\]
Then any tilt in $\calF_\phi$ is constant within cells, and the fitted radius is piecewise constant. A multiscale variant uses several nested grids, for example dyadic partitions. The same idea applies to vector-valued or mixed-type IVs. If a tree or forest trained on $\mathcal I_{\mathrm{tr}}$ partitions $\calZ$ into leaves $L_1,\dots,L_G$, then the leaf indicators define a finite-dimensional map.

\subsection{Sieve and Tensor-Product Constructions}
Let $\psi_1,\psi_2,\dots$ be a sequence of basis functions on $\calZ$. A growing sieve class takes the form
\[
\phi^{(r)}(z)=\p{1,\psi_1(z),\dots,\psi_r(z)}.
\]
Common choices include B-splines, wavelets, Fourier series, and orthogonal polynomials. If one believes that the relevant density ratios are smooth, then the approximation error of the sieve is the natural price of replacing an infinite-dimensional smoothness class by the finite-dimensional span.

For $Z=(Z_1,\dots,Z_{k_Z})$, one may use additive or tensor-product sieves. The additive construction keeps the dimension manageable,
\[
\phi^{(r)}(z)=\Bigp{1,\psi_{1,1}(z_1),\dots,\psi_{1,r_1}(z_1),\dots,\psi_{k_Z,1}(z_{k_Z}),\dots,\psi_{k_Z,r_{k_Z}}(z_{k_Z})}.
\]
Tensor products capture interactions,
\[
\phi_{j_1,\dots,j_q}(z)=\prod_{\ell=1}^q \psi_{\ell,j_\ell}(z_\ell),
\]
but should be used sparingly, because the resulting dimension grows quickly.

\subsection{RKHS Proxies, Nystr\"om Dictionaries, and Random Features}
Let $K$ be a positive-definite kernel on $\calZ$ and $\calH_K$ its RKHS. The ideal smooth shift class is
\[
\calF^{\infty}_{K,B}=\cb{f\in \calH_K: \norm{f}_{\calH_K}\le B,\ f\ge 0}.
\]
Because the main paper develops exact finite-sample theory for finite-dimensional feature maps, the direct implementation of $\calF^{\infty}_{K,B}$ is outside the scope of the exact theorem. A finite-dimensional proxy is obtained by choosing landmarks $u_1,\dots,u_r$ and setting
\[
\phi_r(z)=\p{1,K(z,u_1),\dots,K(z,u_r)}.
\]
A normalized Nystr\"om version replaces the raw kernel sections by a stabilized low-rank basis. For shift-invariant kernels, random features provide another finite-dimensional proxy,
\[
\phi_r(z)=\p{1,\varphi_1(z),\dots,\varphi_r(z)},\qquad \varphi_j(z)=\sqrt{\frac{2}{r}}\cos\p{\omega_j^\top z+b_j}.
\]
Finally, when $Z$ is high-dimensional or structured, one may compose the kernel with a representation map $g$ learned on the training split and use features of the form $K\p{g(z),c_j}$.

\subsection{Shape-Restricted Approximations}
Suppose again that $Z$ is scalar and let $u_1<\cdots<u_G$ be a grid. A piecewise-linear approximation writes
\[
f(z)=\sum_{g=1}^G \theta_g h_g(z),
\]
where $h_1,\dots,h_G$ are hat functions on the grid. Monotone or convex shift families are obtained by adding linear inequality constraints to the coefficients. For example, monotonicity may be encoded by
\[
\theta_1\le \theta_2\le \cdots\le \theta_G,
\]
and convexity by second-difference inequalities. The baseline theorem does not analyze this constrained-cone version, but it remains a natural structured extension.

\subsection{The Separate Radius-Class Design Problem for \texorpdfstring{$\calT_X$}{TX}}
When the final interval is $X$-indexed, one must also choose a class for $\tau_X$. The same approximation ideas apply on the $X$ side, but they should not be conflated with the shift class on the $Z$ side. One may use bins in $X$, additive or tensor-product sieves on $X$, RKHS models on $X$, learned representations of $X$, or shape-restricted classes for $\tau_X$. The shift class $\calF$ determines which probability laws must be protected. The radius class for $\tau_X$ determines how much heterogeneity the prediction interval may have across $X$.

\section{Proofs}\label{appdx:proofs}

\begin{proof}[Proof of Proposition~\ref{prop:moment_char}]
Let
\[
U_\tau=\mathbbm{1}\sqb{Y_{n+1}\in \widehat C_\tau(X_{n+1},Z_{n+1})}-(1-\alpha).
\]
If \eqref{eq:exact_conditional} holds, then $\bbE\sqb{U_\tau\mid Z_{n+1}}=0$ almost surely. Therefore, for every bounded measurable $f$,
\[
\bbE\sqb{f(Z_{n+1})U_\tau}=\bbE\sqb{f(Z_{n+1})\bbE\sqb{U_\tau\mid Z_{n+1}}}=0,
\]
which proves \eqref{eq:moment_char_generic}.

Conversely, suppose \eqref{eq:moment_char_generic} holds for every bounded measurable $f$. Then
\[
\bbE\sqb{f(Z_{n+1})\bbE\sqb{U_\tau\mid Z_{n+1}}}=0,\qquad \forall f\in\calM.
\]
Choosing $f(z)=\mathbbm{1}\sqb{z\in A}$ for an arbitrary measurable set $A\subseteq \calZ$ gives
\[
\bbE\sqb{\mathbbm{1}\sqb{Z_{n+1}\in A}\bbE\sqb{U_\tau\mid Z_{n+1}}}=0,\qquad \forall A.
\]
Hence $\bbE\sqb{U_\tau\mid Z_{n+1}}=0$ almost surely. Rewriting this identity yields
\[
\bbP\p{Y_{n+1}\in \widehat C_\tau(X_{n+1},Z_{n+1})\mid Z_{n+1}}=1-\alpha
\]
almost surely, which is exactly \eqref{eq:exact_conditional}.
\end{proof}

\begin{proof}[Proof of Proposition~\ref{prop:radius_nesting}]
Let $\Gamma(\calT,\calF)$ denote the feasible set of \eqref{eq:oracle_radius_problem} and \eqref{eq:oracle_radius_constraint} within the class $\calT$. If $\calT_1\subset \calT_2$, then every feasible radius in $\Gamma(\calT_1,\calF)$ also belongs to $\Gamma(\calT_2,\calF)$. Therefore the infimum of $\bbE\sqb{\tau(X,Z)}$ over $\Gamma(\calT_2,\calF)$ cannot exceed the corresponding infimum over $\Gamma(\calT_1,\calF)$. This proves
\[
L(\calT_2,\calF)\le L(\calT_1,\calF).
\]
The displayed consequences follow by taking $(\calT_1,\calT_2)=(\calT_Z,\calT_{XZ})$ and $(\calT_1,\calT_2)=(\calT_X,\calT_{XZ})$.
\end{proof}

\begin{proof}[Proof of Theorem~\ref{thm:joint_shift_robust}]
Condition on the training split. Then $\widehat h$ is fixed. Relabel the calibration sample as $\cb{(S_i,W_i)}^m_{i=1}$, where $W_i=(X_i,Z_i)$, and define the test score
\[
S_{m+1}=\abs{Y_{n+1}-\widehat h(X_{n+1})},\qquad W_{m+1}=(X_{n+1},Z_{n+1}).
\]
By assumption, conditional on the training split, the augmented sample $\cb{(S_i,W_i)}^{m+1}_{i=1}$ is exchangeable.

The calibration procedure of Section 4 is exactly the finite-dimensional conditional conformal construction of \citet{Gibbs2025conformalprediction} applied to the conformity score $S=\abs{Y-\widehat h(X)}$ and the covariate $W=(X,Z)$. Their finite-sample result therefore yields the conditional moment inequality
\begin{align}
\label{eq:proof_joint_moment}
\bbE\sqb{f(W_{m+1})\Bigp{\mathbbm{1}\sqb{S_{m+1}\le \widehat\tau_{XZ}(W_{m+1})}-(1-\alpha)}\mid \mathcal I_{\mathrm{tr}}}\ge 0
\end{align}
for every nonnegative $f\in\calF_W$.

Fix any measurable $f\in\calF_W$ with $f\ge 0$ and $\bbE\sqb{f(W)}>0$. Under the tilted law $\bbP^W_f$ defined in \eqref{eq:joint_tilt}, conditional on the training split, we have
\begin{align*}
\bbP_f^W\p{S_{m+1}\le \widehat\tau_{XZ}(W_{m+1})\mid \mathcal I_{\mathrm{tr}}}
&=
\frac{\bbE\sqb{f(W_{m+1})\mathbbm{1}\sqb{S_{m+1}\le \widehat\tau_{XZ}(W_{m+1})}\mid \mathcal I_{\mathrm{tr}}}}{\bbE\sqb{f(W_{m+1})\mid \mathcal I_{\mathrm{tr}}}}\\
&\ge
\frac{(1-\alpha)\bbE\sqb{f(W_{m+1})\mid \mathcal I_{\mathrm{tr}}}}{\bbE\sqb{f(W_{m+1})\mid \mathcal I_{\mathrm{tr}}}}
=1-\alpha,
\end{align*}
where the inequality uses \eqref{eq:proof_joint_moment}.

Finally, by the definition of the interval in \eqref{eq:joint_interval},
\[
S_{m+1}\le \widehat\tau_{XZ}(W_{m+1}) \Longleftrightarrow Y_{n+1}\in \widehat C_{XZ}(X_{n+1},Z_{n+1}).
\]
Hence
\[
\bbP_f^W\p{Y_{n+1}\in \widehat C_{XZ}(X_{n+1},Z_{n+1})\mid \mathcal I_{\mathrm{tr}}}\ge 1-\alpha.
\]
Taking expectations over the training split gives the claimed unconditional statement.
\end{proof}

\begin{proof}[Proof of Theorem~\ref{thm:shift_robust}]
We verify explicitly that the IV-indexed procedure is a specialization of the joint construction. In Section 4, set $W=Z$ and choose the feature map on $W$ to be the IV feature map $\phi$. Then
\[
\calG_W=\cb{w\mapsto \beta^\top\phi(w):\beta\in\bbR^d}=\calG,\qquad \calF_W=\calG_W\cap\cb{f\ge 0}=\calF.
\]
The augmented quantile-regression problem of Section 4 becomes exactly the IV-indexed problem in Section 5, because the covariate $W_i$ is now simply $Z_i$. Consequently the fitted radius from the joint construction satisfies $\widehat\tau_{XZ}(W_{n+1})=\widehat\tau_Z(Z_{n+1})$, and the interval \eqref{eq:joint_interval} reduces to the IV-indexed interval \eqref{eq:iv_indexed_interval}.

Under the same specialization, the tilted law \eqref{eq:joint_tilt} becomes
\[
\frac{d\bbP_f^W(y,x,z)}{d\bbP(y,x,z)}=\frac{f(z)}{\bbE\sqb{f(Z)}},
\]
which is exactly the IV-tilt family \eqref{eq:tilt_family}. Every assumption of Theorem~\ref{thm:joint_shift_robust} is therefore satisfied, and applying that theorem yields
\[
\bbP_f\p{Y_{n+1}\in \widehat C_Z(X_{n+1},Z_{n+1})}\ge 1-\alpha
\]
for every measurable $f\in\calF$ with $f\ge 0$ and $\bbE\sqb{f(Z)}>0$.
\end{proof}

\begin{proof}[Proof of Proposition~\ref{prop:stability}]
Fix $z\in\calZ$ and let
\[
T_z(s_1,\dots,s_m)
\]
denote the IV-indexed radius produced by Algorithm~\ref{alg:ivccp} at IV $z$ when the calibration scores are $s_1,\dots,s_m$ and the calibration IVs are kept fixed. Set
\[
\delta=\norm{\widehat h-h_0}_\infty.
\]
For each calibration observation,
\[
\abs{S_i-S_i^\star}=\abs{\abs{Y_i-\widehat h(X_i)}-\abs{Y_i-h_0(X_i)}}\le \abs{\widehat h(X_i)-h_0(X_i)}\le \delta.
\]
Therefore,
\[
S_i^\star-\delta\le S_i\le S_i^\star+\delta,\qquad i=1,\dots,m.
\]
We next use two properties of $T_z$. First, $T_z$ is coordinatewise monotone in the calibration scores. If one increases some calibration score while keeping everything else fixed, the largest accepted radius cannot decrease. Second, because the constant functions belong to $\calG$, the map is translation-equivariant,
\[
T_z(s_1+c,\dots,s_m+c)=T_z(s_1,\dots,s_m)+c,\qquad c\in\bbR.
\]
Indeed, shifting every calibration score and the candidate test score by the same constant leaves all residuals in the pinball loss unchanged after replacing $g$ by $g+c$, and $g+c\in\calG$ because $\calG$ contains constants.

Applying the coordinatewise inequalities to the monotone map $T_z$ gives
\[
T_z(S_1^\star-\delta,\dots,S_m^\star-\delta)\le T_z(S_1,\dots,S_m)\le T_z(S_1^\star+\delta,\dots,S_m^\star+\delta).
\]
By translation equivariance,
\[
T_z(S_1^\star-\delta,\dots,S_m^\star-\delta)=T_z(S_1^\star,\dots,S_m^\star)-\delta=\widehat\tau_Z^\star(z)-\delta,
\]
\[
T_z(S_1^\star+\delta,\dots,S_m^\star+\delta)=T_z(S_1^\star,\dots,S_m^\star)+\delta=\widehat\tau_Z^\star(z)+\delta.
\]
Hence
\[
\widehat\tau_Z^\star(z)-\delta\le \widehat\tau_Z(z)\le \widehat\tau_Z^\star(z)+\delta,
\]
which is equivalent to
\[
\abs{\widehat\tau_Z(z)-\widehat\tau_Z^\star(z)}\le \delta=\norm{\widehat h-h_0}_\infty.
\]
The interval-length bound follows immediately because $\mathrm{len}(\widehat C_Z(x,z))=2\widehat\tau_Z(z)$ and $\mathrm{len}(\widehat C_Z^\star(x,z))=2\widehat\tau_Z^\star(z)$.
\end{proof}

\begin{proof}[Proof of Theorem~\ref{thm:length}]
Under \eqref{eq:npiv}, the oracle scores are
\[
S_i^\star=\abs{Y_i-h_0(X_i)}=\abs{\varepsilon_i},\qquad i=1,\dots,m.
\]
Fix $z\in\calZ$. For the augmented sample consisting of the $m$ oracle calibration pairs and the additional pair $(z,\widehat\tau_Z^\star(z))$, let $\widehat g_z\in\calG$ be an optimal basic solution of the augmented quantile-regression problem, chosen so that $\widehat g_z(z)=\widehat\tau_Z^\star(z)$. Write the residuals as
\[
r_i=S_i^\star-\widehat g_z(Z_i),\qquad i=1,\dots,m,
\qquad r_{m+1}=\widehat\tau_Z^\star(z)-\widehat g_z(z)=0.
\]
Let
\[
\mathcal A_z=\cb{i\in\cb{1,\dots,m+1}: r_i=0}
\]
be the active set. Because $\widehat g_z$ is chosen from a $d$-dimensional linear class and we take an optimal basic feasible solution of the associated linear program, at most $d$ residual constraints can be tight. Therefore,
\[
\abs{\mathcal A_z}\le d.
\]
The KKT conditions for quantile regression imply the usual empirical quantile-balance relation. Since the only ambiguity in that relation comes from observations with zero residual, the deviation of the empirical quantile level from its target is bounded by the fraction of active residuals. Therefore,
\begin{align}
\label{eq:proof_empirical_balance_new}
\abs{\frac{1}{m+1}\sum_{i=1}^{m+1}\mathbbm{1}\sqb{S_i^\star\le \widehat g_z(Z_i)}-(1-\alpha)}\le \frac{\abs{\mathcal A_z}}{m+1}\le \frac{d}{m+1}.
\end{align}
Because $\widehat g_z(z)=\widehat\tau_Z^\star(z)$, Assumption~\ref{ass:transfer} combines with \eqref{eq:proof_empirical_balance_new} to yield
\[
\abs{F_z\p{\widehat\tau_Z^\star(z)}-(1-\alpha)}\le \frac{d}{m+1}+o_p(1).
\]
Since $F_z\p{\tau_0(z)}=1-\alpha$, we obtain
\[
\abs{F_z\p{\widehat\tau_Z^\star(z)}-F_z\p{\tau_0(z)}}\le \frac{d}{m+1}+o_p(1).
\]
Assumption~\ref{ass:density} then implies
\[
\abs{\widehat\tau_Z^\star(z)-\tau_0(z)}\le \frac{d}{(m+1)p_{\min}}+o_p(1).
\]
The bound for $\widehat\tau_Z(z)$ follows by adding and subtracting $\widehat\tau_Z^\star(z)$ and using Proposition~\ref{prop:stability}:
\begin{align*}
\abs{\widehat\tau_Z(z)-\tau_0(z)}
&\le \abs{\widehat\tau_Z(z)-\widehat\tau_Z^\star(z)}+\abs{\widehat\tau_Z^\star(z)-\tau_0(z)}\\
&\le \norm{\widehat h-h_0}_\infty+\frac{d}{(m+1)p_{\min}}+o_p(1).
\end{align*}
Finally,
\begin{align*}
\mathrm{len}(\widehat C_Z(x,z)) - \mathrm{len}(C_0(x,z))
&=2\widehat\tau_Z(z)-2\tau_0(z)\\
&\le 2\abs{\widehat\tau_Z(z)-\tau_0(z)}\\
&\le 2\norm{\widehat h-h_0}_\infty+\frac{2d}{(m+1)p_{\min}}+o_p(1),
\end{align*}
which completes the proof.
\end{proof}

\begin{proof}[Proof of Proposition~\ref{prop:iw_identity}]
Fix any measurable $\tau_X\colon \calX\to\bbR_+$. By definition of the conditional density ratio,
\begin{align*}
\bbE\sqb{m_{\tau_X}(S,X)r_0(S,X\mid z)}
&=
\int m_{\tau_X}(s,x)\frac{p_{S,X\mid Z}(s,x\mid z)}{p_{S,X}(s,x)}p_{S,X}(s,x)\,ds\,dx\\
&=
\int m_{\tau_X}(s,x)p_{S,X\mid Z}(s,x\mid z)\,ds\,dx\\
&=\bbE\sqb{m_{\tau_X}(S,X)\mid Z=z},
\end{align*}
which is \eqref{eq:iw_identity}. The equivalence with \eqref{eq:xindexed_conditional_moment} is immediate.
\end{proof}

\begin{proof}[Proof of Proposition~\ref{prop:iw_characterization}]
If $Q(\tau_X)=0$, then the integrand
\[
z\mapsto \Bigp{\bbE\sqb{m_{\tau_X}(S,X)r_0(S,X\mid z)}}^2
\]
vanishes for $\nu$-almost every $z$. By Proposition~\ref{prop:iw_identity}, this is equivalent to
\[
\bbE\sqb{m_{\tau_X}(S,X)\mid Z=z}=0
\]
for $\nu$-almost every $z$. The converse implication is immediate. If $\nu$ has full support and the conditional moment is continuous in $z$, then vanishing $\nu$-almost everywhere implies vanishing everywhere on the support of $Z$, which yields exact conditional coverage for $\widehat C_X$.
\end{proof}

\begin{proof}[Proof of Theorem~\ref{thm:weighted_conformal_tx}]
Condition on the samples used to construct $\widehat h$ and $q$. Then $\widehat h$ and $q$ are fixed, so the normalized recalibration scores
\[
R_i=\frac{\abs{Y_i-\widehat h(X_i)}}{q(X_i)},\qquad i\in\mathcal I_{\mathrm{rcal}},
\]
and the test score
\[
R_{n+1}=\frac{\abs{Y_{n+1}-\widehat h(X_{n+1})}}{q(X_{n+1})}
\]
are exchangeable under the training law. The target probability law is the single shifted law $\bbP_{f_0}$ induced by the known weights $w_{f_0}(Z)=f_0(Z)/\bbE\sqb{f_0(Z)}$.

For a test point with IV value $z$, weighted split conformal prediction under covariate shift, as developed by \citet{Tibshirani2019conformalprediction}, applies directly to the normalized score $R$ with covariate $Z$ and target weight $w_{f_0}$. Let
\[
\widehat t_{f_0}(z)=\inf\cb{t\in\bbR_+: \sum_{i\in\mathcal I_{\mathrm{rcal}}}\frac{w_{f_0}(Z_i)}{\sum_{j\in\mathcal I_{\mathrm{rcal}}}w_{f_0}(Z_j)+w_{f_0}(z)}\mathbbm{1}\{R_i\le t\}\ge 1-\alpha},
\]
with the convention $\inf\emptyset=\infty$. Their finite-sample theorem yields
\[
\bbP_{f_0}\p{R_{n+1}\le \widehat t_{f_0}(Z_{n+1})}\ge 1-\alpha.
\]
Because $w_{f_0}(Z_{n+1})\le B_{f_0}$ almost surely and the cutoff is nondecreasing in the test weight, we have
\[
\widehat t_{f_0}(Z_{n+1})\le \widehat t_{f_0}
\qquad\text{almost surely.}
\]
Therefore,
\[
\bbP_{f_0}\p{R_{n+1}\le \widehat t_{f_0}}\ge 1-\alpha.
\]
By definition of the interval \eqref{eq:xindexed_scaled_interval},
\[
R_{n+1}\le \widehat t_{f_0}
\Longleftrightarrow
Y_{n+1}\in \widehat C_{X,f_0}(X_{n+1}).
\]
Hence
\[
\bbP_{f_0}\p{Y_{n+1}\in \widehat C_{X,f_0}(X_{n+1})}\ge 1-\alpha,
\]
which proves the claim.
\end{proof}

\end{document}